\pdfoutput=1
\documentclass{acmsmall}

\pdfoutput=1 
\usepackage[varg]{txfonts}

\let\subcaption\relax

\usepackage[bookmarks=false]{hyperref}
\usepackage{amsmath}
\usepackage{algorithm,algpseudocode}
\usepackage{amsfonts}
\usepackage{booktabs}
\usepackage{graphicx}
\usepackage{color}

\usepackage{url}
\usepackage{paralist}
\usepackage{balance}
\usepackage{ramkidefns}
\usepackage{subcaption}

\usepackage{tikz}
\usepackage{pgfplots}
\usepackage{pgfplotstable}
\usepackage{etoolbox}
\usetikzlibrary{patterns}
\usepackage{longtable}
\usepackage{colortbl} 

\graphicspath{{./figs/}}

\newcommand{\ignore}[1]{}
\newcommand{\nnz}{\text{nnz}}
\newcommand{\NaiveAlg}{Naive-Parallel-AUNMF}
\newcommand{\ParNMF}{MPI-FAUN}

\newcommand{\Naive}{\textbf{Naive}}
\newcommand{\MU}{\textbf{MU}}
\newcommand{\HALS}{\textbf{HALS}}
\newcommand{\BPP}{\textbf{ABPP}}
\newcommand{\LUC}{\textbf{LUC}}

\let\OLDthebibliography\thebibliography
\renewcommand\thebibliography[1]{
  \OLDthebibliography{#1}
  \setlength{\parskip}{0pt}
  \setlength{\itemsep}{0.8 \itemsep}
}

\makeatletter \def\runningfoot{\def\@runningfoot{}} \def\firstfoot{\def\@firstfoot{}} \makeatother 

\acmArticle{}
  
\begin{document}

\setlength{\pdfpageheight}{\paperheight}
\setlength{\pdfpagewidth}{\paperwidth}




\title{MPI-FAUN: An MPI-Based Framework for Alternating-Updating Nonnegative Matrix Factorization}

\author{Ramakrishnan Kannan
\affil{Oak Ridge National Laboratories, TN}
Grey Ballard
\affil{Wake Forest University, NC}
Haesun Park
\affil{Georgia Institute of Technology, GA}
}

\begin{abstract}
Non-negative matrix factorization (NMF) is the problem of determining two non-negative low rank factors $\WW$ and $\HH$, for the given input matrix $\AA$, such that $\AA \approx \WW \HH$.  
NMF is a useful tool for many applications in different domains such as topic modeling in text mining, background separation in video analysis, and community detection in social networks.  
Despite its popularity in the data mining community, there is a lack of efficient parallel algorithms to solve the problem for big data sets.  

The main contribution of this work is a new, high-performance parallel computational framework for a broad class of NMF algorithms that iteratively solves alternating non-negative least squares (NLS) subproblems for $\WW$ and $\HH$. 
It maintains the data and factor matrices in memory (distributed across processors), uses MPI for interprocessor communication, and, in the dense case, provably minimizes communication costs (under mild assumptions).  
The framework is flexible and able to leverage a variety of NMF and NLS algorithms, including Multiplicative Update, Hierarchical Alternating Least Squares, and Block Principal Pivoting.
Our implementation allows us to benchmark and compare different algorithms on massive dense and sparse data matrices of size that spans for few hundreds of millions to billions.  
We demonstrate the scalability of our algorithm and compare it with baseline implementations, showing significant performance improvements. 
The code and the datasets used for conducting the experiments are available online.
\end{abstract}





\maketitle

\section{Introduction}

Non-negative Matrix Factorization (NMF) is the problem of finding two low rank factors $\WW\in \Rnplus{m\times k}$ and $\HH\in \Rnplus{k \times n}$ for a given input matrix  $\AA\in \Rnplus{m\times n}$, such that $\AA \approx \WW \HH$.
Here, $\Rnplus{m\times n}$ denotes the set of $m \times n$ matrices with non-negative real values.
Formally, the NMF problem \cite{seung2001algorithms} can be defined as \SplitN{\label{eqn:original NMF}}{
\min_{\WW \geq 0,\HH \geq 0} & \|\AA-\WW\HH\|_F,
}
where $\|\M{X}\|_F=(\sum_{ij} x_{ij}^2)^{1/2}$ is the Frobenius norm.

NMF is widely used in data mining and machine learning as a dimension reduction and factor analysis method. 
It is a natural fit for many real world problems as the non-negativity is inherent in many representations of real-world data and
the resulting low rank factors are expected to have a natural interpretation. The applications of NMF range from text mining \cite{pauca2004text},  computer vision \cite{hoyer2004non}, and bioinformatics \cite{kim2007sparse} to blind source separation  \cite{cichocki2009nonnegative}, unsupervised clustering \cite{kuang2012symmetric,kuang2013symnmf}  and many other areas.
In the typical case, $k \ll \min(m,n)$; for problems today, $m$ and $n$ can be on the order of millions or more, and $k$ is on the order of few tens to thousands.

There is a vast literature
on algorithms for NMF and their convergence properties \cite{kim2013nonnegative}.   
The commonly adopted NMF algorithms are -- (i) Multiplicative Update (\MU) \cite{seung2001algorithms} (ii) Hierarchical Alternating Least Squares (\HALS) \cite{cichocki2009nonnegative,Ho2008} (iii) NMF based on Alternating Nonnegative Least Squares and Block Principal Pivoting (\BPP) \cite{kim2011fast}, and (iv) Stochastic Gradient Descent (SGD) Updates \cite{gemulla2011large}. 
Most of the algorithms in NMF literature are based on alternately optimizing each of the low rank factors $\WW$ and $\HH$ while keeping the other fixed, in which case each subproblem is a constrained convex optimization problem. 
Subproblems can then be solved using standard optimization techniques such as projected gradient or interior point method; a detailed survey for solving such problems can be found in \cite{xiong2013survey,kim2013nonnegative}. 
In this paper, our implementation uses either \BPP, \MU, or \HALS. 
But our parallel framework is extensible to other algorithms as-is or with a few modifications, as long as they fit an alternating-updating framework (defined in Section \ref{sec:aunmf}).  

With the advent of large scale internet data and interest in Big Data, researchers have started studying scalability of many foundational machine learning algorithms. 
To illustrate the dimension of matrices commonly used in the machine learning community, we present a few examples. 
Nowadays the adjacency matrix of a billion-node social network is common. 
In the matrix representation of a video data, every frame contains three matrices for each RGB color, which is reshaped into a column.  
Thus in the case of a 4K video, every frame will take approximately 27 million rows (4096 row pixels x 2196 column pixels x 3 colors). 
Similarly, the popular representation of documents in text mining is a bag-of-words matrix, where the rows are the dictionary and the columns are the documents (e.g., webpages). 
Each entry $A_{ij}$ in the bag-of-words matrix is generally the frequency count of the word $i$  in the document $j$. 
Typically with the explosion of the new terms in social media, the number of words spans to millions. 
To handle such high-dimensional matrices, it is important to study low-rank approximation methods in a data-distributed and parallel computing environment. 

In this work, we present an efficient algorithm and implementation using tools from the field of High-Performance Computing (HPC).
We maintain data in memory (distributed across processors), take advantage of optimized libraries like BLAS and LAPACK for local computational routines, and use the Message Passing Interface (MPI) standard to organize interprocessor communication.
Furthermore, the current hardware trend is that available parallelism (and therefore aggregate computational rate) is increasing much more quickly than improvements in network bandwidth and latency, which implies that the relative cost of communication (compared to computation) is increasing.
To address this challenge, we analyze algorithms in terms of both their computation and communication costs.
In particular, we prove in Section \ref{sec:parNMF} that in the case of dense input and under a mild assumption, our proposed algorithm minimizes the amount of data communicated between processors to within a constant factor of the lower bound.

A key attribute of our framework is that the efficiency does not require a loss of generality of NMF algorithms.
Our central observation is that most NMF algorithms consist of two main tasks: (a) performing matrix multiplications and (b) solving Non-negative Least Squares (NLS) subproblems, either approximately or exactly.
More importantly, NMF algorithms tend to perform the same matrix multiplications, differing only in how they solve NLS subproblems, and the matrix multiplications often dominate the running time of the algorithms.
Our framework is designed to perform the matrix multiplications efficiently and organize the data so that the NLS subproblems can be solved independently in parallel, leveraging any of a number of possible methods.
We explore the overall efficiency of the framework and compare three different NMF methods in Section \ref{sec:experiment}, performing convergence, scalability, and parameter-tuning experiments on over 1500 processors.

\begin{table}[htp]
\begin{center}
\begin{tabular}{|c|c|c|c|c|}
\hline 
Dataset & Type & Matrix size & NMF Time \\ \hline
Video & Dense & 1 Million x 13,824 &  5.73 seconds \\
Stack Exchange & Sparse & 627,047 x 12 Million &  67 seconds \\
Webbase-2001 & Sparse & 118 Million x 118 Million & 25 minutes \\ \hline
\end{tabular}
\end{center}
\caption{\ParNMF\ on large real-world datasets. Reported time is for 30 iterations on 1536 processors with a low rank of 50.}
\label{tab:teaser}
\end{table}%

With our framework, we are able to explore several large-scale synthetic and real-world data sets, some dense and some sparse. In Table \ref{tab:teaser}, we present the NMF computation wall clock time on some very large real world datasets. We describe the results of the computation in Section \ref{sec:experiment}, showing the range of application of NMF and the ability of our framework to scale to large data sets.  


A preliminary version of this work has already appeared as a conference paper \cite{KBP16}.
While the focus of the previous work was parallel performance of \BPP, the goal of this paper is to explore more data analytic questions.
In particular, the new contributions of this paper include (1) implementing a software framework to compare \BPP\ with \MU\ and \HALS\ for large scale data sets, (2) benchmarking on a data analysis cluster and scaling up to over 1500 processors, and (3) providing an interpretation of results for real-world data sets.
We provide a detailed comparison with other related work, including MapReduce implementations of NMF, in Section \ref{sec:related}.

Our main contribution is a new, high-performance parallel computational framework for a broad class of NMF algorithms. 
The framework is efficient, scalable, flexible, and demonstrated to be effective for large-scale dense and sparse matrices.  
Based on our survey and knowledge, we are the fastest NMF implementation available in the literature.  
The code and the datasets used for conducting the experiments can be downloaded from \url{https://github.com/ramkikannan/nmflibrary}. 
\section{Preliminaries}

\subsection{Notation}
\label{sec:notations}

Table \ref{tab:notation} summarizes the notation we use throughout this paper.
We use \emph{upper case} letters for matrices and \emph{lower case} letters for vectors.
We use both subscripts and superscripts for sub-blocks of matrices. 
For example, $\AA_i$ is the $i$th row block of matrix $\AA$, and $\AA^i$ is the $i$th column block.
Likewise, $\aa_i$ is the $i$th row of $\AA$, and $\aa^i$ is the $i$th column.
We use $m$ and $n$ to denote the numbers of rows and columns of $\AA$, respectively, and we assume without loss of generality $m\geq n$ throughout.

\begin{table}
\begin{center}
\begin{tabular}{|l|l|}
\hline
$\AA$ & Input matrix \\
$\WW$ & Left low rank factor \\
$\HH$ & Right low rank factor \\
$m$ & Number of rows of input matrix \\
$n$ & Number of columns of input matrix \\
$k$ & Low rank \\
$\M{M}_i$ & $i$th row block of matrix $\M{M}$ \\
$\M{M}^i$ & $i$th column block of matrix $\M{M}$  \\
$\M{M}_{ij}$ & $(i,j)$th subblock of $\M{M}$ \\
$p$ & Number of parallel processes \\
$p_r$ & Number of rows in processor grid \\
$p_c$ & Number of columns in processor grid \\
\hline
\end{tabular}
\end{center}
\caption{Notation}
\label{tab:notation}
\end{table}%

\subsection{Communication model}
\label{sec:comm-model}

To analyze our algorithms, we use the $\alpha$-$\beta$-$\gamma$ model of distributed-memory parallel computation.
In this model, interprocessor communication occurs in the form of messages sent between two processors across a bidirectional link (we assume a fully connected network).
We model the cost of a message of size $n$ words as $\alpha+n\beta$, where $\alpha$ is the per-message latency cost and $\beta$ is the per-word bandwidth cost.
Each processor can compute floating point operations (flops) on data that resides in its local memory; $\gamma$ is the per-flop computation cost.
With this communication model, we can predict the performance of an algorithm in terms of the number of flops it performs as well as the number of words and messages it communicates.
For simplicity, we will ignore the possibilities of overlapping computation with communication in our analysis.
For more details on the $\alpha$-$\beta$-$\gamma$ model, see \cite{TRG05,CH+07}.

\subsection{MPI collectives}
\label{sec:collectives}

Point-to-point messages can be organized into collective communication operations that involve more than two processors.
MPI provides an interface to the most commonly used collectives like broadcast, reduce, and gather, as the algorithms for these collectives can be optimized for particular network topologies and processor characteristics.
The algorithms we consider use the all-gather, reduce-scatter, and all-reduce collectives, so we review them here, along with their costs.
Our analysis assumes optimal collective algorithms are used (see \cite{TRG05,CH+07}), though our implementation relies on the underlying MPI implementation.

At the start of an all-gather collective, each of $p$ processors owns data of size $n/p$. 
After the all-gather, each processor owns a copy of the entire data of size $n$. 
The cost of an all-gather is $\alpha\cdot \log p + \beta \cdot \frac{p-1}{p}n$.
At the start of a reduce-scatter collective, each processor owns data of size $n$.
After the reduce-scatter, each processor owns a subset of the sum over all data, which is of size $n/p$.
(Note that the reduction can be computed with other associative operators besides addition.)
The cost of an reduce-scatter is $\alpha\cdot \log p + (\beta+\gamma) \cdot \frac{p-1}{p}n$.
At the start of an all-reduce collective, each processor owns data of size $n$.
After the all-reduce, each processor owns a copy of the sum over all data, which is also of size $n$.
The cost of an all-reduce is $2\alpha\cdot \log p + (2\beta+\gamma) \cdot \frac{p-1}{p}n$.
Note that the costs of each of the collectives are zero when $p=1$.

\section{Related Work}\label{sec:related}

In the data mining and machine learning literature there is an overlap between low rank approximations and matrix factorizations due to the nature of applications. 
Despite its name, non-negative matrix ``factorization'' is really a low rank approximation. 
Recently there is a growing interest in collaborative filtering based 
recommender systems. One of the popular techniques
for collaborative filtering is matrix factorization, often with nonnegativity constraints, 
and its implementation is widely available in many
off-the-shelf distributed machine learning libraries
such as GraphLab \cite{low2012}, MLLib \cite{meng2015mllib},
and many others \cite{satish2014,yun2014} as well.
However, we would like to clarify that collaborative
filtering using matrix factorization is a different problem than NMF: 
in the case of collaborative filtering, non-nonzeros in the matrix 
are considered to be missing entries, while in the case of NMF, non-nonzeros in the matrix correspond to true zero values.
 
There are several recent distributed NMF algorithms in the literature \cite{liao2014cloudnmf,Faloutsos2014,Yin2014,liu2010distributed}. 
Liu et al.\ propose running Multiplicative Update (MU) for KL divergence, squared loss, and ``exponential'' loss functions \cite{liu2010distributed}. 
Matrix multiplication, element-wise multiplication, and element-wise division are the building blocks of the MU algorithm. 
The authors discuss performing these matrix operations effectively in Hadoop for sparse matrices. 
Using similar approaches, Liao et al.\ implement an open source Hadoop-based MU algorithm and study its scalability on large-scale biological data sets \cite{liao2014cloudnmf}. 
Also, Yin, Gao, and Zhang present a scalable NMF that can perform frequent updates, which aim to use the most recently updated data \cite{Yin2014}. 
Similarly Faloutsos et al.\ propose a distributed, scalable method for decomposing matrices, tensors, and coupled data sets through stochastic gradient descent on a variety of objective functions \cite{Faloutsos2014}. 
The authors also provide an implementation that can enforce non-negative constraints on the factor matrices. 
All of these works use Hadoop to implement their algorithms.

We emphasize that our MPI-based approach has several advantages over Hadoop-based approaches:
\begin{itemize}
	\item efficiency -- our approach maintains data in memory, never communicating the data matrix, while Hadoop-based approaches must read/write data to/from disk and involves global shuffles of data matrix entries;
	\item generality -- our approach is well-designed for both dense and sparse data matrices, whereas Hadoop-based approaches generally require sparse inputs;
	\item privacy -- our approach allows processors to collaborate on computing an approximation without ever sharing their local input data (important for applications involving sensitive data, such as electronic health records), while Hadoop requires the user to relinquish control of data placement.
\end{itemize}

We note that Spark \cite{ZCFSS10} is a popular big-data processing infrastructure that is generally more efficient for iterative algorithms such as NMF than Hadoop, as it maintains data in memory and avoids file system I/O.
Even with a Spark implementation of previously proposed Hadoop-based NMF algorithm, we expect performance to suffer from expensive communication of input matrix entries, and Spark will not overcome the shortcomings of generality and privacy of the previous algorithms.
Although Spark has collaborative filtering libraries such as MLlib \cite{meng2015mllib}, which use matrix factorization and can impose non-negativity constraints, none of them implement pure NMF, and so we do not have a direct comparison against NMF running on Spark.
As mentioned above, the problem of collaborative filtering is different from NMF, and therefore different computations are performed at each iteration.

Fairbanks et al. \cite{Fairbanks2015} present a parallel NMF algorithm designed for multicore machines.  
To demonstrate the importance of minimizing communication, we consider this approach to parallelizing an alternating-updating NMF algorithm in distributed memory (see Section \ref{sec:naive}).
While this naive algorithm exploits the natural parallelism available within the alternating iterations (the fact that rows of $\WW$ and columns of $\HH$ can be computed independently), it performs more communication than necessary to set up the independent problems.
We compare the performance of this algorithm with our proposed approach to demonstrate the importance of designing algorithms to minimize communication; that is, simply parallelizing the computation is not sufficient for satisfactory performance and parallel scalability.

Apart from distributed NMF algorithms using Hadoop and multicores, there are also implementations of the
MU algorithm in a distributed memory setting using X10 \cite{Grove2014} and on a GPU \cite{mejia2015nmf}.

\section{Alternating-Updating NMF Algorithms}
\label{sec:aunmf}

We define Alternating-Updating NMF algorithms as those that (1) alternate between updating $\WW$ for a given $\HH$ and updating $\HH$ for a given $\WW$ and (2) use the Gram matrix associated with the fixed factor matrix and the product of the input data matrix $\AA$ with the fixed factor matrix.
We show the structure of the framework in Algorithm \ref{alg:aunmf}. 


\begin{algorithm}
\caption{$[\WW,\HH] = \text{AU-NMF}(A,k)$}
\label{alg:aunmf}
\begin{algorithmic}[1]
\Require $\AA$ is an $m\times n$ matrix, $k$ is rank of approximation
\State Initialize $\HH$ with a non-negative matrix in $\Rn{n\times k}_+$.
\While{stopping criteria not satisfied} \label{algo:nmfloop}
  \State Update $\WW$ using $\HH \HH^T$ and $\AA \HH^T$
  \label{line:aunmf:W}
  \State Update $\HH$ using $\WW^T\WW$ and $\WW^T \AA$
  \label{line:aunmf:H}
\EndWhile
\end{algorithmic}
\end{algorithm}

The specifics of lines \ref{line:aunmf:W} and  \ref{line:aunmf:H} depend on the NMF algorithm, and we refer to the computation associated with these lines as the Local Update Computations (\LUC), as they will not affect the parallelization schemes we define in Section \ref{sec:parNMF}.
Because these computations are performed locally, we use a function $F(m,n,k)$ to denote the number of flops required for each algorithm's \LUC\ (and we do not consider communication costs).

We note that AU-NMF is very similar to a two-block, block coordinate descent (BCD) framework, but it has a key difference.
In the BCD framework where the two blocks are
the unknown factors $\WW$ and $\HH$, 
we \emph{solve} the following subproblems,
which have a unique solution for a full rank $\HH$ and $\WW$: 
\SplitN{\label{eqn:two block}} {
\WW &\leftarrow \Argmin{\tilde \WW\geq 0}\NormBr{\AA-\tilde\WW\HH}_F,\\
\HH &\leftarrow  \Argmin{\tilde\HH\geq 0}\NormBr{\AA-\WW\tilde\HH}_F.
}
Since each subproblem involves nonnegative least squares,
this two-block BCD method is also called
the Alternating Non-negative Least Squares (ANLS) method \cite{kim2013nonnegative}.
For example, Block Principal Pivoting (\BPP), discussed more in detail at Section \ref{sec:BPP}, is 
one algorithm that solves these NLS subproblems.
In the context of the AU-NMF algorithm,
 an ANLS method {\em maximally} reduces 
the overall NMF objective function value 
by finding the optimal solution for
 given $\HH$ and $\WW$ in lines \ref{line:aunmf:W} 
and \ref{line:aunmf:H} respectively.  

There are other popular NMF algorithms 
that update the factor matrices alternatively
without maximally reducing the objective function value each time,
in the same sense as in ANLS. 
These updates do not necessarily solve each of the subproblems \eqref{eqn:two block} to optimality but simply improve the overall objective function \eqref{eqn:original NMF}.  
Such methods include Multiplicative Update (\MU) \cite{seung2001algorithms} and Hierarchical Alternating Least Squares (\HALS) \cite{cichocki2009nonnegative}, which was also proposed as Rank-one Residual Iteration (RRI) \cite{Ho2008}.
To show how these methods can fit into the AU-NMF framework, we discuss them in more detail in Sections \ref{sec:MU} and \ref{sec:HALS}.

The convergence properties of these different algorithms are discussed in detail by Kim, He and Park \cite{kim2013nonnegative}. 
We emphasize here that both \MU\ and \HALS\ require computing Gram matrices and matrix products of the input matrix and each factor matrix.
Therefore, if the update ordering follows the convention of updating all of $\WW$ followed by all of $\HH$, both methods fit into the AU-NMF framework. 
We note that both \MU\ and \HALS\ are defined for more general update orders, but for our purposes we constrain them to be AU-NMF algorithms.

While we focus on three NMF algorithms in this paper, we highlight that our framework is extensible to other NMF algorithms, including those based on Alternating Direction Method of Multipliers (ADMM) \cite{SF14}, Nesterov-based methods \cite{GTLY12}, or any other method that fits the framework of Algorithm \ref{alg:aunmf}.

\subsection{Multiplicative Update (\MU)}
\label{sec:MU}

In the case of \MU\ \cite{seung2001algorithms}, individual entries of $\WW$ and $\HH$ are updated with all other entries fixed.
In this case, the update rules are 
\SplitN{\label{eqn:muupdate}} {
w_{ij} &\leftarrow w_{ij} \frac{(\AA \HH^T)_{ij}}{(\WW \HH \HH^T)_{ij}}, \text{ and }\\
h_{ij} &\leftarrow  h_{ij} \frac{(\WW^T \AA)_{ij}}{(\WW^T \WW \HH)_{ij}}.
} 
Instead of performing these $(m+n)k$ in an arbitrary order, if all of $\WW$ is updated before $\HH$ (or vice-versa), this method also follows the AU-NMF framework.
After computing the Gram matrices $\HH\HH^T$ and $\WW^T \WW$ and the products $\AA\HH^T$ and $\WW^T\AA$, the extra cost of computing $\WW (\HH\HH^T)$ and $(\WW^T\WW)\HH$ is $F(m,n,k)=2(m+n)k^2$ flops to perform updates for all entries of $\WW$ and $\HH$, as the other elementwise operations affect only lower-order terms.
Thus, when \MU\ is used, lines \ref{line:aunmf:W} and \ref{line:aunmf:H} in Algorithm \ref{alg:aunmf} -- and functions UpdateW and UpdateH in Algorithms \ref{alg:naive} and \ref{alg:2D} -- implement the expressions in \eqref{eqn:muupdate}, given the previously computed matrices.

\subsection{Hierarchical Alternating Least Squares (\HALS)}
\label{sec:HALS}

In the case of \HALS\ \cite{cichocki2009nonnegative,CA2009}, updates are performed on individual columns of $\WW$ and rows of $\HH$ with all other entries in the factor matrices fixed.
This approach is a BCD method with $2k$ blocks, set to minimize the function
\begin{equation}
f(\ww^1,\cdots,\ww^k,\hh_1,\cdots,\hh_k)=\lt\|\AA-\sum_{i=1}^k \ww^i \hh_i \rt\|_F,\label{eq:hals_obj}
\end{equation}
where $\ww^i$ is the $i$th column of $\WW$ and $\hh_i$ is the $i$th row of $\HH$.
The update rules \cite[Algorithm 2]{CA2009} can be written in closed form:
\SplitN{\label{eqn:halsupdate}} {
\ww^i &\leftarrow \lt[ \ww^i + (\AA\HH^T)^i - \WW (\HH \HH^T)^i \rt]_+ \\
\ww^i &\leftarrow \frac{\ww^i}{\|\ww^i\|}, \text{ and } \\
\hh_i &\leftarrow \lt[ \hh_i + (\WW^T\AA)_i - (\WW^T \WW)_i\HH \rt]_+.
} 

Note that the columns of $\WW$ and rows of $\HH$ are updated in order, so that the most up-to-date values are always used, and these $2k$ updates can be done in an arbitrary order.  However, if all the $\WW$ updates are done before $\HH$ (or vice-versa),  the method falls into the AU-NMF framework.
After computing the matrices $\HH\HH^T$, $\AA\HH^T$, $\WW^T\WW$, and $\WW^T\AA$, the extra computation is $F(m,n,k)=2(m+n)k^2$ flops for updating both $\WW$ and $\HH$. 

Thus, when \HALS\ is used, lines \ref{line:aunmf:W} and \ref{line:aunmf:H} in Algorithm \ref{alg:aunmf} -- and functions UpdateW and UpdateH in Algorithms \ref{alg:naive} and \ref{alg:2D} -- implement the expressions in \eqref{eqn:halsupdate}, given the previously computed matrices.  




\subsection{Alternating Nonnegative Least Squares with Block Principal Pivoting}
\label{sec:BPP}


Block Principal Pivoting (BPP) is an active-set-like method for solving the NLS subproblems in Eq. \eqref{eqn:two block}.
The main subroutine of BPP is the single right-hand side NLS problem
\SplitN{\label{eqn:single NLS}}{
\min_{\xx\geq 0} \|\CC\xx-\mathbf{b}\|_2.
}

The Karush-Kuhn-Tucker (KKT) optimality conditions for  Eq.~\eqref{eqn:single NLS} are as follows
\SplitS{\label{eqn:KKT}}{
\yy &= \CC^T \CC \xx - \CC^T \mathbf{b} \\
\yy &\geq 0\\
\xx &\geq 0\\
x_i y_i & = 0 \;\; \forall i .
}
The KKT conditions \eqref{eqn:KKT} states that at optimality, the support sets (i.e., the non-zero elements)
of $\xx$ and $\yy$ are complementary to each other. Therefore, Eq.~\eqref{eqn:KKT} is an instance of
the \emph{Linear Complementarity Problem} (LCP) which arises frequently in quadratic programming.
When $k\ll\min(m,n)$, active-set and active-set-like methods are very suitable because most
computations involve matrices of sizes $m\times k, n\times k$, and $k\times k$ which are
small and easy to handle.

If we knew which indices correspond to nonzero values in the optimal solution, then computing the solution is an unconstrained least squares problem on these indices.
In the optimal solution, call the set of indices $i$ such that $x_i=0$ the active set, and let the remaining indices be the passive set. The BPP algorithm works to find this final active set and passive set. 
It greedily swaps indices between the intermediate active and passive sets until finding a partition that satisfies the KKT condition. In the partition of the optimal solution, the values of the indices that belong to the active set will take zero. The values of the indices that belong to the passive set are determined by solving the unconstrained least squares problem restricted to the passive set. Kim, He and Park \cite{kim2011fast}, discuss the BPP algorithm in further detail. 
We use the notation
$$\XX \gets \text{SolveBPP}(\CC^T\CC,\CC^T\BB)$$
to define the (local) function for using BPP to solve Eq.~\eqref{eqn:single NLS} for every column of $\XX$.
We define $C_\text{BPP}(k,c)$ as the cost of \text{SolveBPP}, given the $k\times k$ matrix $\CC^T\CC$ and $k\times c$ matrix $\CC^T\BB$. 
\text{SolveBPP} mainly involves solving least squares problems over the intermediate passive sets. 
Our implementation uses the normal equations to solve the unconstrained least squares problems because the normal equations matrices have been pre-computed in order to check the KKT condition.
However, more numerically stable methods such as QR decomposition can also be used.

Thus, when \BPP\ is used, lines \ref{line:aunmf:W} and \ref{line:aunmf:H} in Algorithm \ref{alg:aunmf} -- and functions UpdateW and UpdateH in Algorithms \ref{alg:naive} and \ref{alg:2D} -- correspond to calls to SolveBPP.
The number of flops involved in SolveBPP is not a closed form expression; in this case $F(m,n,k)=C_\text{BPP}(k,m)+ C_\text{BPP}(k,n)$.

\section{Parallel Algorithms}
\label{sec:parallel}

\subsection{Naive Parallel NMF Algorithm} 
\label{sec:naive}

In this section we present a naive parallelization of NMF algorithms, which has previously appeared in the context of a shared-memory parallel platform \cite{Fairbanks2015}. 
Each NLS problem with multiple right-hand sides can be parallelized on the observation that the problems for multiple right-hand sides are independent from each other. 
For example, we can solve several instances of Eq.~\eqref{eqn:single NLS} independently for different $\mathbf{b}$ where $\CC$ is fixed, which implies that we can optimize row blocks of $\WW$ and column blocks of $\HH$ in parallel. 

\begin{algorithm}[t!]
\caption{$[\WW,\HH] = \text{\NaiveAlg}(\AA,k)$}
\label{alg:naive}
\begin{algorithmic}[1]
\Require $\AA$ is an $m\times n$ matrix distributed both row-wise and column-wise across $p$ processors, $k$ is rank of approximation
\Require Local matrices: $\AA_{i}$ is $m/p\times n$, $\AA^{i}$ is $m\times n/p$, $\WW_i$ is $m/p\times k$, $\HH^i$ is $k\times n/p$
\State $p_i$ initializes $\HH^i$
\While{stopping criteria not satisfied}
	\Statex \quad\; \textbf{/* Compute $\WW$ given $\HH$ */} 
	\State collect $\HH$ on each processor using all-gather
		\label{line:naive:allgatherH}
	\State $p_i$ computes $\WW_i \gets \text{updateW}(\HH\HH^T,\AA_i\HH^T)$
		\label{line:naive:computeW}
	\Statex \quad\; \textbf{/* Compute $\HH$ given $\WW$ */} 
	\State collect $\WW$ on each processor using all-gather
		\label{line:naive:allgatherW}
	\State $p_i$ computes $(\HH^i)^T \gets \text{updateH}(\WW^T\WW,(\WW^T\AA^i)^T)$
		\label{line:naive:computeH}
\EndWhile
\Ensure $\displaystyle \WW, \HH \approx \Argmin{\M{\tilde W} \geq 0, \M{\tilde H} \geq 0} \|\AA- \M{\tilde W} \M{\tilde H}\|$
\Ensure $\WW$ is an $m\times k$ matrix distributed row-wise across processors, $\HH$ is a $k\times n$ matrix distributed column-wise across processors
\end{algorithmic}
\end{algorithm}

\begin{figure}[t!]
\centering
\includegraphics[width=\textwidth]{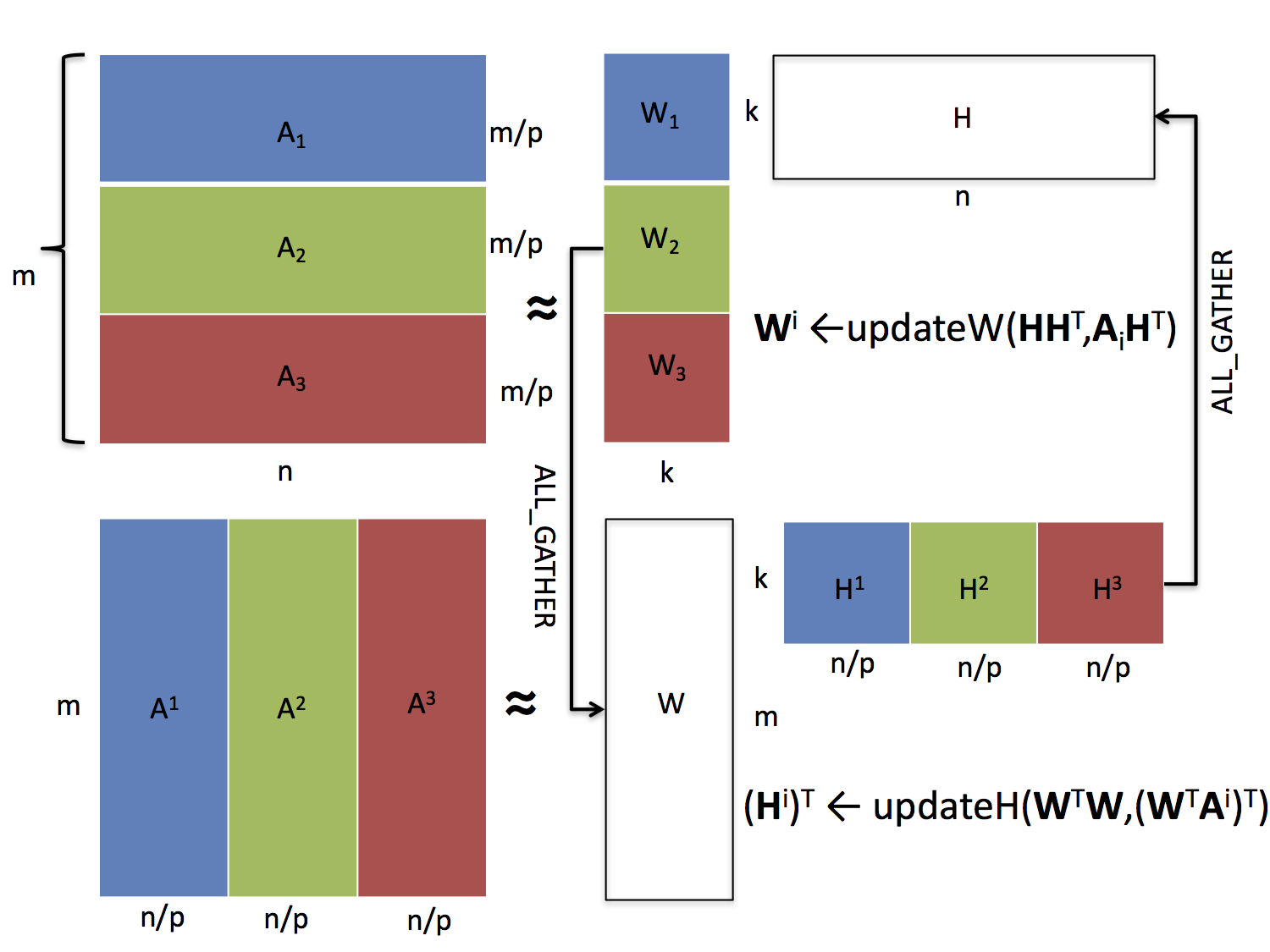}
\caption{\NaiveAlg. Note that both rows and columns of $A$ are 1D distributed.  The algorithm works by iteratively (all-)gathering the entire fixed factor matrix to each processor and then performing the Local Update Computations to update the variable factor matrix.} 
\label{fig:naive}
\end{figure}

\renewcommand{\arraystretch}{1.25}
\newcommand{\smaller}{\scriptsize}
\begin{table*}[t!]
\begin{center}
\begin{tabular}{|c|c|c|c|c|}
\hline
\textbf{Algorithm} & \textbf{Flops} & \textbf{Words} & \textbf{Messages} & \textbf{Memory}  \\ \hline
\smaller \NaiveAlg & \smaller $4\frac{mnk}{p}+(m{+}n)k^2+F\lt(\frac mp,\frac np,k\rt)$ & \smaller $O((m+n)k)$ & \smaller $O(\log p)^*$ & \smaller $O\lt(\frac{mn}{p}+(m{+}n)k\rt)$ \\ \hline
\smaller \ParNMF\ ($m/p \geq n$) & \smaller $4\frac{mnk}{p}+\frac{(m+n)k^2}{p} + F\lt(\frac mp,\frac np,k\rt)$ & \smaller $O(nk)$ & \smaller $O(\log p)^*$ & \smaller $O\lt(\frac{mn}{p}+\frac{mk}{p}+nk\rt)$ \\ \hline
\smaller \ParNMF\ ($m/p < n$) & \smaller $4\frac{mnk}{p}+\frac{(m+n)k^2}{p}+F\lt(\frac mp,\frac np,k\rt)$ & \smaller $O\lt( \sqrt{\frac{mnk^2}{p}}\rt)$ & \smaller $O(\log p)^*$ & \smaller $O\lt(\frac{mn}{p}+\sqrt{\frac{mnk^2}{p}}\rt)$ \\ \hline
\smaller Lower Bound & $-$ & \smaller $\Omega\lt(\min\lt\{\sqrt{\frac{mnk^2}{p}},nk\rt\}\rt)$ & \smaller $\Omega(\log p)$ & \smaller $\frac{mn}{p}+\frac{(m+n)k}{p}$ \\ \hline
\end{tabular}
\normalsize
\end{center}
\caption{Leading order algorithmic costs for \NaiveAlg{} and \ParNMF{} (per iteration).  Note that the computation and memory costs assume the data matrix $\AA$ is dense, but the communication costs (words and messages) apply to both dense and sparse cases.
The function $F(\cdot)$ denotes the number of flops required for the particular NMF algorithm's Local Update Computation, aside from the matrix multiplications common across AU-NMF algorithms. \\
$^*$The stated latency cost assumes no communication is required in \LUC; \HALS\ requires $k\log p$ messages for normalization steps.}
\label{tab:costs}
\end{table*}%

Algorithm \ref{alg:naive} and Figure \ref{fig:naive} present a straightforward approach to setting up the independent subproblems.
Let us divide $\WW$ into row blocks $\WW_1, \ldots, \WW_p$ and $\HH$ into column blocks $\HH^1, \ldots, \HH^p$. 
We then double-partition the data matrix $\AA$ accordingly into row blocks $\AA_{1}, \ldots, \AA_p$ and column blocks $\AA^1, \ldots, \AA^p$ so that processor $i$ owns both $\AA_i$ and $\AA^i$ (see Figure \ref{fig:naive}).
With these partitions of the data and the variables, one can implement any AU-NMF algorithm in parallel, with only one communication step for each solve.

We summarize the algorithmic costs of Algorithm \ref{alg:naive} (derived in the following subsections) in Table \ref{tab:costs}.
This naive algorithm \cite{Fairbanks2015} has three main drawbacks: (1) it requires storing two copies of the data matrix (one in row distribution and one in column distribution) and both full factor matrices locally, (2) it does not parallelize the computation of $\HH\HH^T$ and $\WW^T\WW$ (each processor computes it redundantly), and (3) as we will see in Section \ref{sec:parNMF}, it communicates more data than necessary.

\subsubsection{Computation Cost}

The computation cost of Algorithm \ref{alg:naive} depends on the particular NMF algorithm used.
Thus, the computation at line \ref{line:naive:computeW} consists of computing $\AA^i\HH^T$, $\HH\HH^T$, and performing the algorithm-specific Local Update Computations for $m/p$ rows of $\WW$.
Likewise, the computation at line \ref{line:naive:computeH} consists of computing $\WW^T\AA_i$, $\WW^T\WW$, and performing the Local Update Computations for $n/p$ columns of $\HH$.
In the dense case, this amounts to $4mnk/p+(m+n)k^2+F(m/p,n/p,k)$ flops.
In the sparse case, processor $i$ performs $2(\nnz(\AA_i)+\nnz(\AA^i))k$ flops to compute $\AA^i\HH^T$ and $\WW^T\AA_i$ instead of $4mnk/p$. 

\subsubsection{Communication Cost}

The size of $\WW$ is $mk$ words, and the size of $\HH$ is $nk$ words.
Thus, the communication cost of the all-gathers at lines \ref{line:naive:allgatherH} and \ref{line:naive:allgatherW}, based on the expression given in Section \ref{sec:collectives} is $\alpha\cdot 2\log p + \beta\cdot (m+n)k$.

\subsubsection{Memory Requirements} \label{sec:memory}

The local memory requirement includes storing each processor's part of matrices $\AA$, $\WW$, and $\HH$.
In the case of dense $\AA$, this is $2mn/p+(m+n)k/p$ words, as $\AA$ is stored twice; in the sparse case, processor $i$ requires $\nnz(\AA_i)+\nnz(\AA^i)$ words for the input matrix and $(m+n)k/p$ words for the output factor matrices.
Local memory is also required for storing temporary matrices $\WW$ and $\HH$ of size $(m+n)k$ words.

\subsection{\ParNMF}
\label{sec:parNMF}

We present our proposed algorithm, \ParNMF, as Algorithm \ref{alg:2D}. 
The main ideas of the algorithm are to (1) exploit the independence of Local Update Computations for rows of $\WW$ and columns of $\HH$ and (2) use communication-optimal matrix multiplication algorithms to set up the Local Update Computations.
The naive approach (Algorithm \ref{alg:naive}) shares the first property, by parallelizing over rows of $\WW$ and columns of $\HH$, but it uses parallel matrix multiplication algorithms that communicate more data than necessary.
The central intuition for communication-efficient parallel algorithms for computing $\HH\HH^T$, $\AA\HH^T$, $\WW^T\WW$, and $\WW^T\AA$ comes from a classification proposed by Demmel et al. \cite{DE+13}.
They consider three cases, depending on the relative sizes of the dimensions of the matrices and the number of processors; the four multiplies for NMF fall into either the ``one large dimension'' or ``two large dimensions" cases.
\ParNMF\ uses a careful data distribution in order to use a communication-optimal algorithm for each of the matrix multiplications, while at the same time exploiting the parallelism in the \LUC.

The algorithm uses a 2D distribution of the data matrix $\AA$ across a $p_r \times p_c$ grid of processors (with $p=p_rp_c$), as shown in Figure \ref{fig:2D-distribution}.
As we derive in the subsequent subsections, Algorithm \ref{alg:2D} performs an alternating method in parallel with a per-iteration bandwidth cost of $O\lt(\min\lt\{\sqrt{mnk^2/p},nk\rt\}\rt)$ words, latency cost of $O(\log p)$ messages, and load-balanced computation (up to the sparsity pattern of $\AA$ and convergence rates of local BPP computations).

To minimize the communication cost and local memory requirements, in the typical case $p_r$ and $p_c$ are chosen so that $m/p_r\approx n/p_c\approx \sqrt{mn/p}$, in which case the bandwidth cost is $O\lt(\sqrt{mnk^2/p}\rt)$.
If the matrix is very tall and skinny, i.e., $m/p>n$, then we choose $p_r=p$ and $p_c=1$.
In this case, the distribution of the data matrix is 1D, and the bandwidth cost is $O(nk)$ words.

The matrix distributions for Algorithm \ref{alg:2D} are given in Figure \ref{fig:2D-distribution}; we use a 2D distribution of $\AA$ and 1D distributions of $\WW$ and $\HH$.
Recall from Table \ref{tab:notation} that $\M{M}_i$ and $\M{M}^i$  denote row and column blocks of $\M{M}$, respectively.
Thus, the notation $(\WW_i)_j$ denotes the $j$th row block within the $i$th row block of $\WW$.
Lines \ref{line:2DsyrkH}--\ref{line:2DcompW} compute $\WW$ for a fixed $\HH$, and lines \ref{line:2DsyrkW}--\ref{line:2DcompH} compute $\HH$ for a fixed $\WW$; note that the computations and communication patterns for the two alternating iterations are analogous.

In the rest of this section, we derive the per-iteration computation and communication costs, as well as the local memory requirements.
We also argue the communication-optimality of the algorithm in the dense case.
Table \ref{tab:costs} summarizes the results of this section and compares them to \NaiveAlg.

\begin{algorithm}[t!]
\caption{$[\WW,\HH] = \text{\ParNMF}(\AA,k)$}
\label{alg:2D}
\begin{algorithmic}[1]
\small
\Require $\AA$ is an $m\times n$ matrix distributed across a $p_r\times p_c$ grid of processors, $k$ is rank of approximation
\Require Local matrices: $\AA_{ij}$ is $m/p_r\times n/p_c$, $\WW_i$ is $m/p_r\times k$, $(\WW_i)_j$ is $m/p\times k$, $\HH_j$ is $k\times n/p_c$, and $(\HH_j)_i$ is $k\times n/p$
\State $p_{ij}$ initializes $(\HH_j)_i$
\While{stopping criteria not satisfied}
	\Statex \quad\; \textbf{/* Compute $\WW$ given $\HH$ */} 
	\State $p_{ij}$ computes $\M{U}_{ij}=(\HH_j)_i{(\HH_j)_i}^T$
		\label{line:2DsyrkH}
	\State compute $\HH\HH^T {=} \sum_{i,j} \M{U}_{ij}$ using all-reduce across all procs
		\label{line:2Dall-reduceH}
		\Comment{$\HH\HH^T$ is $k\times k$ and symmetric}
	\State $p_{ij}$ collects $\HH_j$ using all-gather across proc columns
		\label{line:2Dall-gatherH}
	\State $p_{ij}$ computes $\M{V}_{ij}=\AA_{ij}\HH_j^T$
		\label{line:2DNEW}
		\Comment{$\M{V}_{ij}$ is $m/p_r \times k$}
	\State compute $(\AA\HH^T)_i {=} \sum_j \M{V}_{ij}$ using reduce-scatter across proc row to achieve row-wise distribution of $(\AA\HH^T)_i$
		\Comment{$p_{ij}$ owns $m/p\times k$ submatrix $((\AA\HH^T)_i)_j$}
		\label{line:2Dreduce-scatterAHT}
	\State $p_{ij}$ computes $(\WW_i)_j \gets \text{UpdateW}(\HH\HH^T,((\AA\HH^T)_i)_j)$
		\label{line:2DcompW}
	\Statex \quad\; \textbf{/* Compute $\HH$ given $\WW$ */}
	\State $p_{ij}$ computes $\M{X}_{ij}={(\WW_i)_j}^T(\WW_i)_j$
		\label{line:2DsyrkW}
	\State compute $\WW^T\WW {=} \sum_{i,j} \M{X}_{ij}$ using all-reduce across all procs
		\label{line:2Dall-reduceW}
		\Comment{$\WW^T\WW$ is $k\times k$ and symmetric}
	\State $p_{ij}$ collects $\WW_i$ using all-gather across proc rows
		\label{line:2Dall-gatherW}
	\State $p_{ij}$ computes $\M{Y}_{ij}={\WW_i}^T\AA_{ij}$
		\label{line:2DNEH}
		\Comment{$\M{Y}_{ij}$ is $k\times n/p_c$}
	\State compute $(\WW^T\AA)^j = \sum_i \M{Y}_{ij}$ using reduce-scatter across proc columns to achieve column-wise distribution of $(\WW^T\AA)^j$
		\Comment{$p_{ij}$ owns $k\times n/p$ submatrix $((\WW^T\AA)^j)^i$}
		\label{line:2Dreduce-scatterWTA}
	\State $p_{ij}$ computes $((\HH^j)^i)^T \gets \text{UpdateH}(\WW^T\WW,(((\WW^T\AA)^j)^i)^T)$
		\label{line:2DcompH}
\EndWhile
\Ensure $\displaystyle \WW, \HH \approx \Argmin{\M{\tilde W} \geq 0, \M{\tilde H} \geq 0} \|\AA- \M{\tilde W} \M{\tilde H}\|$
\Ensure $\WW$ is an $m\times k$ matrix distributed row-wise across processors, $\HH$ is a $k\times n$ matrix distributed column-wise across processors
\normalsize
\end{algorithmic}
\end{algorithm}

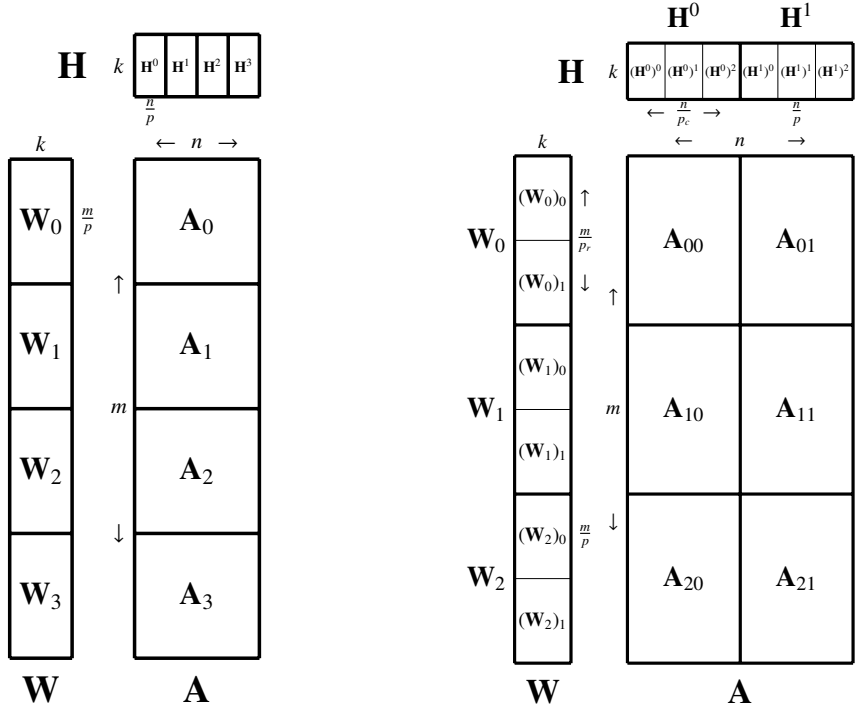
\begin{figure}[t!]
\centering
\begin{subfigure}[b]{.5\textwidth}
\centering
\scalebox{.83}{\begin{tikzpicture}

\draw[xscale=2,yscale=2,ultra thick] (0,0) grid (1,4);
\draw[yscale=2,ultra thick] (-2,0) grid (-1,4);
\draw[xscale=1/2,ultra thick] (0,9) grid (4,10);

\node[draw=none] at (1,-0.5) {\LARGE $\AA$};
\node[draw=none] at (1,7) {\Large $\AA_0$};
\node[draw=none] at (1,5) {\Large $\AA_1$};
\node[draw=none] at (1,3) {\Large $\AA_2$};
\node[draw=none] at (1,1) {\Large $\AA_3$};
\node[draw=none] at (-1.5,-0.5) {\LARGE $\WW$};
\node[draw=none] at (-1.5,7) {\Large $\WW_0$};
\node[draw=none] at (-1.5,5) {\Large $\WW_1$};
\node[draw=none] at (-1.5,3) {\Large $\WW_2$};
\node[draw=none] at (-1.5,1) {\Large $\WW_3$};
\node[draw=none] at (-1,9.5) {\LARGE $\HH$};
\node[draw=none] at (0.25,9.5) {\scriptsize $\HH^0$};
\node[draw=none] at (0.75,9.5) {\scriptsize $\HH^1$};
\node[draw=none] at (1.25,9.5) {\scriptsize $\HH^2$};
\node[draw=none] at (1.75,9.5) {\scriptsize $\HH^3$};

\node[draw=none] at (-0.25,9.5) {$k$};
\node[draw=none] at (-0.25,4) {$m$};
\node[draw=none] at (-0.25,6) {$\uparrow$};
\node[draw=none] at (-0.25,2) {$\downarrow$};
\node[draw=none] at (-0.75,7) {$\frac{m}{p}$};
\node[draw=none] at (-1.5,8.25) {$k$};
\node[draw=none] at (1,8.25) {$n$};
\node[draw=none] at (0.5,8.25) {$\leftarrow$};
\node[draw=none] at (1.5,8.25) {$\rightarrow$};
\node[draw=none] at (0.25,8.75) {$\frac{n}{p}$};

\end{tikzpicture}}
\subcaption{1D Distribution with $p=p_r=4$ and $p_c=1$.}
\label{1D}
\end{subfigure}%
\begin{subfigure}[b]{.5\textwidth}
\centering
\scalebox{.75}{\begin{tikzpicture}

\draw[xscale=2,yscale=3,ultra thick] (0,0) grid (2,3);
\draw[yscale=3/2] (-2,0) grid (-1,6);
\draw[yscale=3,ultra thick] (-2,0) grid (-1,3);
\draw[xscale=2/3] (0,10) grid (6,11);
\draw[xscale=2,ultra thick] (0,10) grid (2,11);

\node[draw=none] at (2,-0.5) {\LARGE $\AA$};
\node[draw=none] at (1,7.5) {\Large $\AA_{00}$};
\node[draw=none] at (1,4.5) {\Large $\AA_{10}$};
\node[draw=none] at (1,1.5) {\Large $\AA_{20}$};
\node[draw=none] at (3,7.5) {\Large $\AA_{01}$};
\node[draw=none] at (3,4.5) {\Large $\AA_{11}$};
\node[draw=none] at (3,1.5) {\Large $\AA_{21}$};
\node[draw=none] at (-1.5,-0.5) {\LARGE $\WW$};
\node[draw=none] at (-2.5,7.5) {\Large $\WW_0$};
\node[draw=none] at (-2.5,4.5) {\Large $\WW_1$};
\node[draw=none] at (-2.5,1.5) {\Large $\WW_2$};
\node[draw=none] at (-1.5,8.25) {$(\WW_0)_0$};
\node[draw=none] at (-1.5,6.75) {$(\WW_0)_1$};
\node[draw=none] at (-1.5,5.25) {$(\WW_1)_0$};
\node[draw=none] at (-1.5,3.75) {$(\WW_1)_1$};
\node[draw=none] at (-1.5,2.25) {$(\WW_2)_0$};
\node[draw=none] at (-1.5,0.75) {$(\WW_2)_1$};
\node[draw=none] at (-1,10.5) {\LARGE $\HH$};
\node[draw=none] at (1,11.5) {\Large $\HH^0$};
\node[draw=none] at (3,11.5) {\Large $\HH^1$};
\node[draw=none] at (.33,10.5) {\scriptsize $(\HH^0)^0$};
\node[draw=none] at (1,10.5) {\scriptsize $(\HH^0)^1$};
\node[draw=none] at (1.66,10.5) {\scriptsize $(\HH^0)^2$};
\node[draw=none] at (2.33,10.5) {\scriptsize $(\HH^1)^0$};
\node[draw=none] at (3,10.5) {\scriptsize $(\HH^1)^1$};
\node[draw=none] at (3.66,10.5) {\scriptsize $(\HH^1)^2$};

\node[draw=none] at (-0.25,10.5) {$k$};
\node[draw=none] at (-0.25,4.5) {$m$};
\node[draw=none] at (-0.25,6.5) {$\uparrow$};
\node[draw=none] at (-0.25,2.5) {$\downarrow$};
\node[draw=none] at (-0.75,7.5) {$\frac{m}{p_r}$};
\node[draw=none] at (-0.75,8.25) {$\uparrow$};
\node[draw=none] at (-0.75,6.75) {$\downarrow$};
\node[draw=none] at (-0.75,2.25) {$\frac{m}{p}$};
\node[draw=none] at (-1.5,9.25) {$k$};
\node[draw=none] at (2,9.25) {$n$};
\node[draw=none] at (1,9.25) {$\leftarrow$};
\node[draw=none] at (3,9.25) {$\rightarrow$};
\node[draw=none] at (1,9.75) {$\frac{n}{p_c}$};
\node[draw=none] at (0.5,9.75) {$\leftarrow$};
\node[draw=none] at (1.5,9.75) {$\rightarrow$};
\node[draw=none] at (3,9.75) {$\frac{n}{p}$};

\end{tikzpicture}}
\subcaption{2D Distribution with $p_r=3$ and $p_c=2$.}
\label{2D}
\end{subfigure}
\caption{Data distributions for \ParNMF. 
Note that for the 2D distribution, $\AA_{ij}$ is $m/p_r \times m/p_c$, $\WW_i$ is $m/p_r \times k$, $(\WW_i)_j$ is $m/p\times k$, $\HH_j$ is $k\times n/p_c$, and $(\HH^j)^i$ is $k\times n/p$.}
\label{fig:2D-distribution}
\end{figure}

\begin{figure}[t!]
\centering
\includegraphics[width=\textwidth]{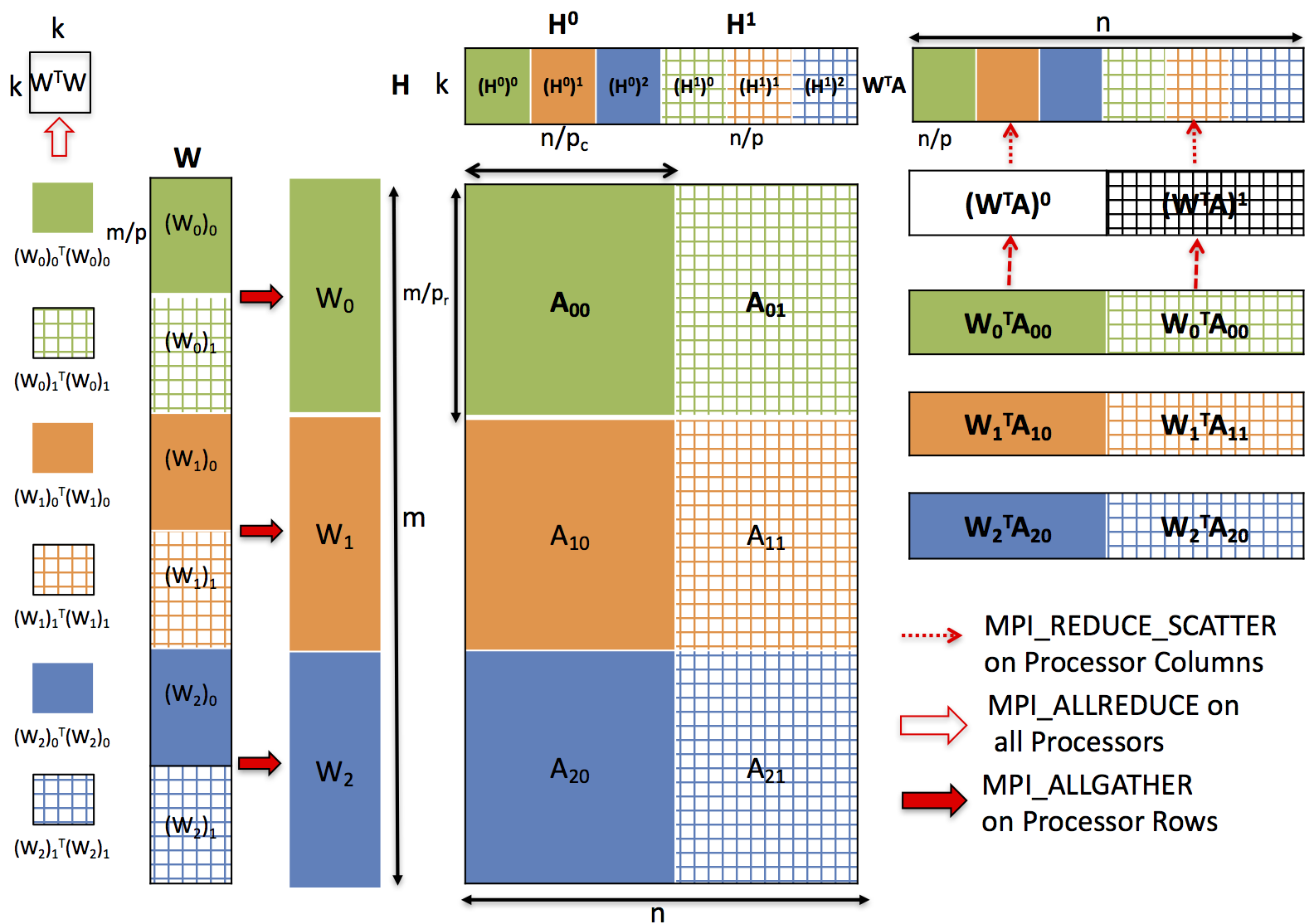}
\caption{Parallel matrix multiplications within \ParNMF\ for finding $\HH$ given $\WW$, with $p_r=3$ and $p_c=2$.  The computation of $\WW^T\WW$ appears on the far left; the rest of the figure depicts computation of $\WW^T\AA$.}
\label{fig:1D-Algorithm}
\end{figure}

\subsubsection{Computation Cost}

Local matrix computations occur at lines \ref{line:2DsyrkH}, \ref{line:2DNEW}, \ref{line:2DsyrkW}, and \ref{line:2DNEH}.
In the case that $\AA$ is dense, each processor performs 
$$\frac np k^2+2\frac{m}{p_r}\frac{n}{p_c}k+\frac mp k^2+2\frac{m}{p_r}\frac{n}{p_c}k=4\frac{mnk}{p}+\frac{(m+n)k^2}{p}$$
flops.
In the case that $\AA$ is sparse, processor $(i,j)$ performs $(m+n)k^2/p$ flops in computing $\M{U}_{ij}$ and $\M{X}_{ij}$, and $4\nnz(\AA_{ij})k$ flops in computing $\M{V}_{ij}$ and $\M{Y}_{ij}$.
Local update computations occur at lines \ref{line:2DcompW} and \ref{line:2DcompH}.
In each case, the symmetric positive semi-definite matrix is $k\times k$ and the number of columns/rows of length $k$ to be computed are $m/p$ and $n/p$, respectively.
These costs together are given by $F(m/p,n/p,k)$.
There are computation costs associated with the all-reduce and reduce-scatter collectives, both those contribute only to lower order terms.

\subsubsection{Communication Cost}
\label{sec:alg:comm}

Communication occurs during six collective operations (lines \ref{line:2Dall-reduceH}, \ref{line:2Dall-gatherH}, \ref{line:2Dreduce-scatterAHT}, \ref{line:2Dall-reduceW}, \ref{line:2Dall-gatherW}, and \ref{line:2Dreduce-scatterWTA}).
We use the cost expressions presented in Section \ref{sec:collectives} for these collectives.
The communication cost of the all-reduces (lines \ref{line:2Dall-reduceH} and \ref{line:2Dall-reduceW}) is $\alpha \cdot 4\log p+\beta \cdot 2k^2$; 
the cost of the two all-gathers (lines \ref{line:2Dall-gatherH} and \ref{line:2Dall-gatherW}) is $\alpha \cdot \log p + \beta \cdot \lt( (p_r{-}1)nk/p + (p_c{-}1)mk/p\rt)$; and
the cost of the two reduce-scatters (lines \ref{line:2Dreduce-scatterAHT} and  \ref{line:2Dreduce-scatterWTA}) is $\alpha \cdot \log p + \beta \cdot \lt( (p_c{-}1)mk/p + (p_r{-}1)nk/p\rt)$.

We note that \LUC\ may introduce significant communication cost, depending on the NMF algorithm used.
The normalization of columns of $\WW$ within \HALS, for example, introduces an extra $k\log p$ latency cost.
We will ignore such costs in our general analysis.

In the case that $m/p<n$, we choose $p_r=\sqrt{mp/n} >1$ and $p_c=\sqrt{np/m}>1$, and these communication costs simplify to $\alpha \cdot O(\log p) + \beta \cdot O(mk/p_r+nk/p_c+k^2) = \alpha \cdot O(\log p) + \beta \cdot O(\sqrt{mnk^2/p}+k^2)$.
In the case that $m/p\geq n$, we choose $p_c=1$, and the costs simplify to $\alpha \cdot O(\log p) + \beta \cdot O(nk)$.

\ignore{
Thus, the total cost per iteration is
$$\gamma \cdot O\lt( \frac{mnk}{p} + C_\text{BPP}\lt(k,\frac{m+n}{p}\rt) \rt) + \beta \cdot O\lt(\frac{mk}{p_r}+\frac{nk}{p_c}+k^2\rt) + \alpha \cdot O(\log p).$$
In order to minimize the bandwidth cost of the algorithm, we choose $p_r=\sqrt{mp/n}$ and $p_c=\sqrt{np/m}$; this yields a per-iteration cost of 
$$\gamma \cdot O\lt( \frac{mnk}{p} + C_\text{BPP}\lt(k,\frac{m+n}{p}\rt) \rt) + \beta \cdot O\lt(\sqrt{\frac{mnk^2}{p}}+k^2\rt) + \alpha \cdot O(\log p).$$
We note that these choices of $p_r$ and $p_c$ are well defined assuming $m/p<n$ and $n/p<m$; if one of these assumptions is not satisfied we resort to Algorithm \ref{alg:1D}.
}

\subsubsection{Memory Requirements}

The local memory requirement includes storing each processor's part of matrices $\AA$, $\WW$, and $\HH$.
In the case of dense $\AA$, this is $mn/p+(m+n)k/p$ words; in the sparse case, processor $(i,j)$ requires $\nnz(\AA_{ij})$ words for the input matrix and $(m+n)k/p$ words for the output factor matrices.
Local memory is also required for storing temporary matrices $\WW_j$, $\HH_i$, $\M{V}_{ij}$, and $\M{Y}_{ij}$, of size $2mk/p_r+2nk/p_c)$ words.

In the dense case, assuming $k<n/p_c$ and $k<m/p_r$, the local memory requirement is no more than a constant times the size of the original data.
For the optimal choices of $p_r$ and $p_c$, this assumption simplifies to $k<\max\lt\{\sqrt{mn/p},m/p\rt\}$.

We note that if the temporary memory requirements become prohibitive, the computation of $((\AA \HH^T)_i)_j$ and $((\WW^T\AA)_j)_i$ via all-gathers and reduce-scatters can be blocked, decreasing the local memory requirements at the expense of greater latency costs.
When $\AA$ is sparse and $k$ is large enough, the memory footprint of the factor matrices can be larger than the input matrix.
In this case, the extra temporary memory requirements can become prohibitive; we observed this for a sparse data set with very large dimensions (see Section \ref{sec:webbase-2001}).
We leave the implementation of the blocked algorithm to future work.

\subsubsection{Communication Optimality}

In the case that $\AA$ is dense, Algorithm \ref{alg:2D} provably minimizes communication costs.
Theorem \ref{thm:LB} establishes the bandwidth cost lower bound for any algorithm that computes $\WW^T\AA$ or $\AA\HH^T$ each iteration.
A latency lower bound of $\Omega(\log p)$ exists in our communication model for any algorithm that aggregates global information \cite{CH+07}, and for NMF, this global aggregation is necessary in each iteration.
Based on the costs derived above, \ParNMF\ is communication optimal under the assumption $k<\sqrt{mn/p}$, matching these lower bounds to within constant factors.

\begin{theorem}[\cite{DE+13}]
\label{thm:LB}
Let $\AA \in \Rn{m \times n}$, $\WW \in \Rn{m \times k}$, and $\HH \in \Rn{k \times n}$ be dense matrices, with $k<n\leq m$.  If $k < \sqrt{mn/p}$, then any distributed-memory parallel algorithm on $p$ processors that load balances the matrix distributions and computes $\WW^T \AA$ and/or $\AA \HH^T$ must communicate at least $\Omega(\min\{\sqrt{mnk^2/p},nk\})$ words along its critical path.
\end{theorem}
\begin{proof}
The proof follows directly from \cite[Section II.B]{DE+13}.
Each matrix multiplication $\WW^T \AA$ and $\AA \HH^T$ has dimensions $k<n\leq m$, so the assumption $k<\sqrt{mn/p}$ ensures that neither multiplication has ``3 large dimensions.''
Thus, the communication lower bound is either $\Omega(\sqrt{mnk^2/p})$ in the case of $p>m/n$ (or ``2 large dimensions''), or $\Omega(nk)$, in the case of $p<m/n$ (or ``1 large dimension'').
If $p<m/n$, then $nk<\sqrt{mnk^2/p}$, so the lower bound can be written as $\Omega(\min\{\sqrt{mnk^2/p},nk\})$.
\end{proof}

We note that the communication costs of Algorithm \ref{alg:2D} are the same for dense and sparse data matrices (the data matrix itself is never communicated).
In the case that $\AA$ is sparse, this communication lower bound does not necessarily apply, as the required data movement depends on the sparsity pattern of $\AA$.
Thus, we cannot make claims of optimality in the sparse case (for general $\AA$).
The communication lower bounds for $\WW^T \AA$ and/or $\AA \HH^T$ (where $\AA$ is sparse) can be expressed in terms of hypergraphs that encode the sparsity structure of $\AA$ \cite{BDKS15}.
Indeed, hypergraph partitioners have been used to reduce communication and achieve load balance for a similar problem: computing a low-rank representation of a sparse tensor (without non-negativity constraints on the factors) \cite{KU15}.

\section{Experiments}
\label{sec:experiment}
\newcommand{\NLS}{LUC }

In this section, we describe our implementation of \ParNMF\ and evaluate its performance.
We identify a few synthetic and real world data sets to experiment with \ParNMF\ with dimensions that span from hundreds to millions. 
We compare the performance and exploring scaling behavior of different NMF algorithms -- \MU, \HALS, and ANLS/BPP (\BPP), implemented using the parallel \ParNMF\ framework.  
The code and the  datasets used for conducting the experiments can be downloaded from \url{https://github.com/ramkikannan/nmflibrary}. 

\subsection{Experimental Setup}

\subsubsection{Data Sets}\label{sec:datasets}

We used sparse and dense matrices that are either synthetically generated or from real world applications. We explain the data sets in this section.

\begin{itemize}
\item Dense Synthetic Matrix: We generate a low rank matrix as the product of two uniform random matrices of size 207,360 $\times$ 100 and 100 $\times$ 138,240. 
The dimensions of this matrix are chosen to be evenly divisible for a particular set of processor grids.  
\item Sparse Synthetic Matrix: We generate a random sparse Erd\H{o}s-R\'{e}nyi matrix of the size 207,360 $\times$ 138,240 with density of 0.001.  That is, every entry is nonzero with probability 0.001.
\item Dense Real World Matrix ({\em Video}):  NMF is used on video data for background subtraction in order to detect moving objects. The low rank matrix $\hat{\AA} = \WW \HH$ represents background and the error matrix $\AA - \hat{\AA}$ represents moving objects.  Detecting moving objects has many real-world applications such as traffic estimation \cite{Fujimoto2014} and security monitoring \cite{BSJJZ2015}.
In the case of detecting moving objects, only the last minute or two of video is taken from the live video camera. The algorithm to incrementally adjust the NMF based on the new streaming video is presented in \cite{kim2013nonnegative}. To simulate this scenario, we collected a video in a busy intersection of the Georgia Tech campus at 20 frames per second. From this video, we took video for approximately 12 minutes and  then reshaped the matrix such that every RGB frame is a column of our matrix, so that the matrix is dense with size 1,013,400 $\times$ 13,824.  
\item Sparse Real World Matrix ({\em Webbase}): This data set is a directed sparse graph whose nodes correspond to webpages (URLs) and edges correspond to hyperlinks from one webpage to another.
The NMF output of this directed graph helps us understand clusters in graphs.
We consider two versions of the data set: {\em webbase-1M} and {\em webbase-2001}.
The dataset webbase-1M contains about 1 million nodes (1,000,005) and 3.1 million edges (3,105,536), and was first reported by Williams et al. \cite{Williams2009}.  
The version webbase-2001 has about 118 million nodes (118,142,155) and over 1 billion edges (1,019,903,190); it was first reported by Boldi and Vigna  \cite{Boldi2004}.  
Both data sets are available in the University of Florida Sparse Matrix Collection \cite{DH11} and the latter {\em webbase-2001} being the largest among the entire collection.
\item Text data ({\em Stack Exchange}): 
Stack Exchange is a network of question-and-answer websites on topics in varied fields, each site covering a specific topic, where questions, answers, and users are subject to a reputation award process. There are many Stack Exchange forums, such as {\em ask ubuntu, mathematics, latex}. 
We downloaded the latest anonymized dump of all 
user-contributed content on the Stack Exchange network from \url{https://archive.org/details/stackexchange} as of 28-Jul-2016. 
We used only the questions from the most popular site called Stackoverflow and did not include the answers and comments. 
We removed the standard 571 English stop words (such as {\em are, am, be, above, below}) and then used snowball stemming available through the Natural Language Toolkit (NLTK) package (\url{www.nltk.org}).
After this initial pre-processing, we deleted HTML tags (such as {\em lt, gt, em}) from the posts. 
The resulting bag-of-words matrix has a vocabulary of size 627,047 over 11,708,841 documents with 365,168,945 non-zero entries. 
\end{itemize}

The size of all the real world data sets were adjusted to the nearest size for uniformly distributing the matrix. 

\subsubsection{Implementation Platform}

We conducted our experiments on ``Rhea'' at the Oak Ridge Leadership Computing Facility (OLCF).
Rhea is a commodity-type Linux cluster with a total of 512 nodes and a 4X FDR Infiniband interconnect.
Each node contains dual-socket 8-core Intel Sandy Bridge-EP processors and 128 GB of memory.
Each socket has a shared 20MB L3 cache, and each core has a private 256K L2 cache. 


Our objective of the implementation is using open source software as much as possible 
to promote reproducibility and reuse of our code.
The entire C++ code was developed using the matrix library Armadillo \cite{sanderson2010}. 
In Armadillo, the elements of the dense matrix are stored in column major order and the sparse matrices in Compressed Sparse Column (CSC) format.
For dense BLAS and LAPACK operations, we linked Armadillo with Intel MKL -- the default LAPACK/BLAS library in RHEA. It is also easy to link Armadillo with OpenBLAS \cite{xianyi2015}. 
We use Armadillo's own implementation of sparse matrix-dense matrix multiplication, the default GNU C++ Compiler (g++ (GCC) 4.8.2) and MPI library (Open MPI 1.8.4)  on RHEA.  We chose the commodity cluster with open source software so that the numbers presented here are representative of common use. 


\subsubsection{Algorithms}

In our experiments, we considered the following algorithms: 
\begin{itemize}
	\item \MU: \ParNMF\ (Algorithm \ref{alg:2D}) with MU (Equation \eqref{eqn:muupdate})
	\item \HALS: \ParNMF\ (Algorithm \ref{alg:2D}) with HALS (Equation \eqref{eqn:halsupdate})
	\item \BPP: \ParNMF\ (Algorithm \ref{alg:2D}) with BPP (Section \ref{sec:BPP})
	\item \Naive: \NaiveAlg\ (Algorithm \ref{alg:naive}, Section \ref{sec:naive})
\end{itemize}

Our implementation of \Naive\ (Algorithm \ref{alg:naive}) uses BPP but can be easily to extended to \MU\ and \HALS\ and other NMF algorithms. 
A detailed comparison of \NaiveAlg\ with \ParNMF\ is made in our earlier work \cite{KBP16}. 
We include some benchmark results from \Naive\ to reiterate the point that communication efficiency is key to obtaining reasonable performance, but we also omit other \Naive\ results in order to focus attention on comparisons among other algorithms.

For the algorithms based on \ParNMF, we use the processor grid that is closest to the theoretical optimum (see Section \ref{sec:alg:comm}) in order to minimize communication costs.
See Section \ref{sec:procgrid} for an empirical evaluation of varying processor grids for a particular algorithm and data set.


To ensure fair comparison among algorithms, the same random seed is used across different methods appropriately. 
That is, the initial random matrix $\HH$ is generated with the same random seed when testing with 
different algorithms (note that $\WW$ need not be initialized). 
In our experiments, we use number of iterations as the stopping criteria for all the algorithms.

While we would like to compare against other high-performance NMF algorithms in the literature, the only other distributed-memory implementations of which we're aware are implemented using Hadoop and are designed only for sparse matrices \cite{liao2014cloudnmf},
\cite{liu2010distributed}, \cite{gemulla2011large}, \cite{Yin2014} and \cite{Faloutsos2014}.
We stress that Hadoop is not designed for high performance computing of iterative numerical 
algorithms, requiring disk I/O between steps, so a run time comparison between a Hadoop 
implementation and a C++/MPI implementation is not a fair comparison of parallel algorithms.
A qualitative example of differences in run time is that a Hadoop implementation of the MU algorithm on 
a large sparse matrix of size $2^{17} \times 2^{16}$ with $2 \times {10^8}$ nonzeros (with k=8) 
takes on the order of 50 minutes per iteration \cite{liu2010distributed}, while our MU implementation 
takes 0.065 seconds per iteration for the synthetic data set (which is an order of magnitude larger in 
terms of rows, columns, and nonzeros) running on only 16 nodes. 

\newcommand{\datafile}{}
\newcommand{\numiterations}{30}
\newcommand{\minvalue}{1}

\newif\ifnaive
\newif\ifksweep

\newcommand{\setcolors}{
\pgfplotsset{cycle list={
	red, fill=red \\ 
	blue, fill=blue \\ 
	green, fill=green \\ 
	red, pattern=crosshatch, pattern color=red \\
	blue, pattern=crosshatch, pattern color=blue \\
	green, pattern=crosshatch, pattern color=green \\
}};
}

\newcommand{\plotoptions}{
	ybar stacked,
	reverse legend,
	bar width=8pt,
	width=11cm, height=3.85cm,
	ylabel={Time (seconds)}, 
	y label style={yshift=-.25cm},
	ymin=0,
	\ifnaive 
		\ifksweep
			symbolic x coords={10-0,10-1,10-2,10-3,,20-0,20-1,20-2,20-3,,30-0,30-1,30-2,30-3,,40-0,40-1,40-2,40-3,,50-0,50-1,50-2,50-3},
		\else
			symbolic x coords={16-0,16-1,16-2,16-3,,96-0,96-1,96-2,96-3,,384-0,384-1,384-2,384-3,,864-0,864-1,864-2,864-3,,1536-0,1536-1,1536-2,1536-3},
		\fi
		xticklabels={MU,HALS,ABPP,Naive,MU,HALS,ABPP,Naive,MU,HALS,ABPP,Naive,MU,HALS,ABPP,Naive,MU,HALS,ABPP,Naive},
	\else  
		\ifksweep
			symbolic x coords={10-0,10-1,10-2,,,20-0,20-1,20-2,,,30-0,30-1,30-2,,,40-0,40-1,40-2,,,50-0,50-1,50-2,},
		\else
			symbolic x coords={16-0,16-1,16-2,,,96-0,96-1,96-2,,,384-0,384-1,384-2,,,864-0,864-1,864-2,,,1536-0,1536-1,1536-2,}, 
		\fi
		xticklabels={MU,HALS,ABPP,MU,HALS,ABPP,MU,HALS,ABPP,MU,HALS,ABPP,MU,HALS,ABPP},
	\fi
	xtick=data,
	xticklabel style={xshift=.15cm,rotate=45,anchor=east},
	\ifksweep
		xlabel={Low Rank ($k$)}, 
	\else
		xlabel={Number of Processes ($p$)},
	\fi
	xlabel style={xshift=1cm,yshift=-0.5cm},
	legend style={draw=none,row sep=-0.1cm},
	legend style={at={(1,.5)},anchor=west}
}

\newcommand{\makeplot}{
\begin{axis}[\plotoptions]
	\setcolors
	\addplot table[x=algo, y expr=(\thisrow{mm}/(\minvalue*\numiterations))] {\datafile};
	\addplot table[x=algo, y expr=(\thisrow{nnls}/(\minvalue*\numiterations))] {\datafile};
	\addplot table[x=algo, y expr=(\thisrow{gram}/(\minvalue*\numiterations))] {\datafile};
	\addplot table[x=algo, y expr=(\thisrow{allgather}/(\minvalue*\numiterations))] {\datafile};
	\addplot table[x=algo, y expr=(\thisrow{reducescatter}/(\minvalue*\numiterations))] {\datafile};
	\addplot table[x=algo, y expr=(\thisrow{allreduce}/(\minvalue*\numiterations))] {\datafile};
	\legend{MM, \NLS, Gram, All-Gather, Reduce-Scatter, All-Reduce};
\end{axis}
}

\newcommand{\labels}{
\node [align=center,text width=3cm] at (1.25cm, -.95cm)   {\ifksweep 10 \else 16 \fi};
\node [align=center,text width=3cm] at (3cm, -0.95cm)   {\ifksweep 20 \else 96 \fi};
\node [align=center,text width=3cm] at (4.6cm, -0.95cm)   {\ifksweep 30 \else 384 \fi};
\node [align=center,text width=3cm] at (6.3cm, -0.95cm) {\ifksweep 40 \else 864 \fi};
\node [align=center,text width=3cm] at (8.15cm, -0.95cm) {\ifksweep 50 \else 1536 \fi};
}

\newcommand{\relerrplot}{
\begin{axis}[xlabel=Iterations, ylabel=Rel.~Error for $k{=}50$,width=3in, height=2.4in]
\addplot [green,very thick] table [x={Iterations}, y={MU}] {\datafile};
\addplot [red,very thick] table [x={Iterations}, y={HALS}] {\datafile};
\addplot [blue,very thick] table [x={Iterations}, y={ANLS-BPP}] {\datafile};
\legend{MU,HALS,ABPP}
\end{axis}
}

\subsection{Relative Error over Iterations} \label{sec:convergence}

There are various metrics to compare the quality of the 
NMF algorithms \cite{kim2013nonnegative}. The most common among these metrics are (a) relative error and (b) projected 
gradient. The former represents the closeness of the low rank approximation $\hat{\AA}\approx\WW\HH$, which is generally the optimization objective. 
The latter 
represent the quality of the produced low rank factors and the stationarity of the final solution. These 
metrics are also used as the stopping criterion for terminating the iteration of the NMF algorithm as in 
line \ref{algo:nmfloop} of Algorithm \ref{alg:aunmf}. Typically a combination of the number of iterations 
along with improvement of these metrics until a tolerance is met is be used as stopping criterion. In this paper, we use 
relative error for the comparison as it is monotonically decreasing, as opposed to projected gradient of the 
low rank factors, which shows oscillations over iterations. The relative error can be formally defined as 
$\|\AA-\WW\HH\|_F/\|\AA\|_F$. 

In Figure \ref{fig:convergence}, we measure the relative error at the end of every iteration (i.e., after the updates of both $\WW$ and $\HH$) for all three algorithms \MU, \HALS, and 
\BPP.
We consider three real world datasets, \emph{video}, \emph{stack exchange} and \emph{webbase-1M}, and set $k=50$. 
We used only the number of iterations as stopping criterion and 
just for this section, ran all the algorithms for 50 iterations. 

To begin with, we explain the observations on the dense \emph{video} dataset presented in Figure 
\ref{fig:denserwerr}. The relative error of \MU\ was  highest at 0.1804 after 50 iterations and \BPP\ was 
the least with  0.1170. \HALS's relative error was 0.1208. From the figure, we can observe that \BPP\ 
error didn't change after 29 iterations where as \HALS\ and \MU\ was still improving marginally at the 
4th decimal even after 50 iterations.  

We can observe that the relative error of \emph{stack exchange} from Figure \ref{fig:stackexchangeerr} is better 
than \emph{webbase-1M} from Figure \ref{fig:sparserwerr} over all three algorithms. 
In the case of the \emph{stack exchange} dataset, the relative errors after 50 iterations follow the pattern \MU\ $>$ \HALS\ $>$ \BPP, with values 0.8480, 0.8365, and 0.8333 respectively. 
Unlike the \emph{video} dataset, both \MU\ and \HALS\ stopped improving after 23 iterations, 
where as \BPP\ was still improving in the 4th decimal even though its error was better than the others. 
However, the difference in relative error for the \emph{webbase-1M} dataset was not as significant as in the others, though the relative ordering of \MU\ $>$ \HALS\ $>$ \BPP\ was consistent, with values of 0.9703 for \MU\, 0.9697 for \HALS\ and 0.9695 for \BPP. 

In general, for these datasets \BPP\ identified better approximations than \MU\ and \HALS, which is consistent with the 
literature \cite{kim2013nonnegative,kim2011fast}. 
However, for the sparse datasets, the differences in relative error are small across the NMF algorithms. 


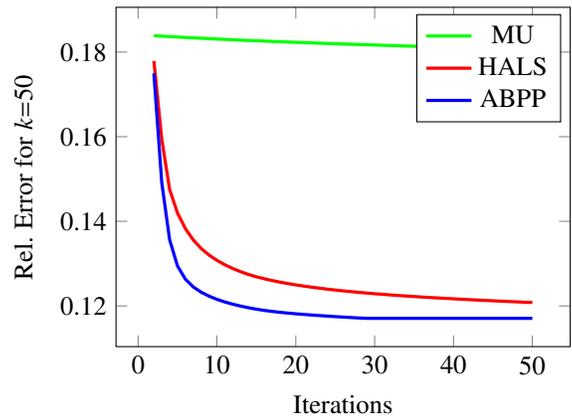
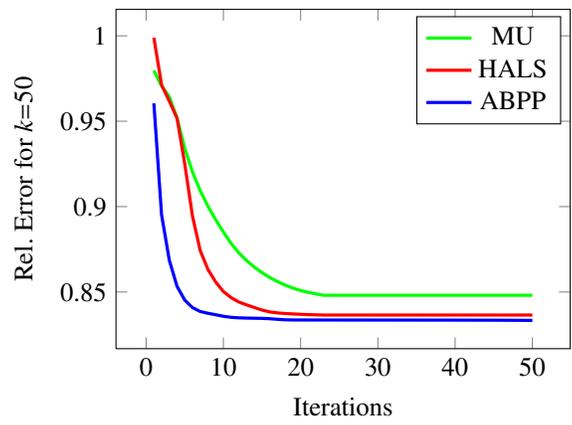
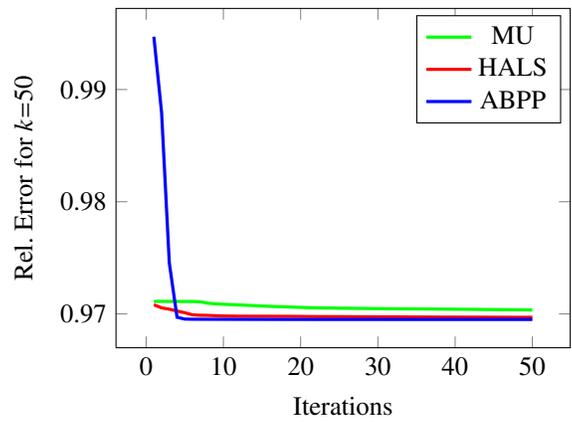
\begin{figure}

\begin{subfigure}[b]{\textwidth}
\centering
\begin{tikzpicture}
\renewcommand{\datafile}{denserwerr.dat}
\relerrplot
\end{tikzpicture}
\subcaption{Dense Real World}
\label{fig:denserwerr}
\end{subfigure}

\begin{subfigure}[b]{\textwidth}
\centering
\begin{tikzpicture}
\renewcommand{\datafile}{stackexchangeerr.dat}
\relerrplot
\end{tikzpicture}
\subcaption{Stack Exchange}
\label{fig:stackexchangeerr}
\end{subfigure}

\begin{subfigure}[b]{\textwidth}
\centering
\begin{tikzpicture}
\renewcommand{\datafile}{sparserwerr.dat}
\relerrplot
\end{tikzpicture}
\subcaption{Webbase}
\label{fig:sparserwerr}
\end{subfigure}

\caption{Relative error comparison of \MU, \HALS, \BPP\ on real world datasets}
\label{fig:convergence}
\end{figure}

\subsection{Time Per Iteration}

In this section we focus on per-iteration time of all the algorithms.
We report four types of experiments, varying the number of processors (Section \ref{sec:scaling}), the rank of the approximation (Section \ref{sec:ksweep}), the shape of the processor grid (Section \ref{sec:procgrid}), and scaling up the dataset size.
For each experiment we report a time breakdown in terms of the overall computation and communication steps (described in Section \ref{sec:perf-breakdown}) shared by all algorithms.

\subsubsection{Time Breakdown}
\label{sec:perf-breakdown}

To differentiate the computation and communication costs among the algorithms, we present the time breakdown among the various tasks within the algorithms for all performance experiments.
For Algorithm \ref{alg:2D}, there are three local computation tasks and three communication tasks to compute each of the factor matrices:
\begin{itemize}
	\item \textbf{MM}, computing a matrix multiplication with the local data matrix and one of the factor matrices;
	\item \textbf{\NLS}, local updates either using \BPP\ or applying the remaining work of the \MU\ or \HALS\ updates (i.e., the total time for both $UpdateW$ and $UpdateH$ functions); 
	\item \textbf{Gram}, computing the local contribution to the Gram matrix;
	\item \textbf{All-Gather}, to compute the global matrix multiplication;
	\item \textbf{Reduce-Scatter}, to compute the global matrix multiplication;
	\item \textbf{All-Reduce}, to compute the global Gram matrix.
\end{itemize}
In our results, we do not distinguish the costs of these tasks for $\WW$ and $\HH$ separately; we report their sum, though we note that we do not always expect balance between the two contributions for each task.
Algorithm \ref{alg:naive} performs all of these tasks except Reduce-Scatter and All-Reduce; all of its communication is in All-Gather.

\subsubsection{Scaling \texorpdfstring{$p$}{p}: Strong Scaling}
\label{sec:scaling}

Figure \ref{fig:scaling} presents a strong scaling experiment with four data sets: \emph{sparse synthetic}, \emph{dense synthetic}, \emph{webbase-1M}, and \emph{video}.
In this experiment, for each data set and algorithm, we use low rank $k=50$ and vary the number of processors (with fixed problem size).
We use $\{1,6,24,54,96\}$ nodes; since each node has 16 cores, this corresponds to $\{16,96,384,864,1536\}$ cores and report average per-iteration times.

We highlight three main observations from these experiments:
\begin{enumerate}
	\item \label{obs:1} \Naive\ is slower than all other algorithms for large $p$;
	\item \label{obs:2} \MU, \HALS, and \BPP\ (algorithms based on \ParNMF) scale up to over 1000 processors;
	\item \label{obs:3} the relative per-iteration cost of \NLS\ decreases as $p$ increases (for all algorithms), and therefore the extra per-iteration cost of \BPP\ (compared with \MU\ and \HALS) becomes negligible.
\end{enumerate}

\paragraph{Observation \ref{obs:1}} 
We report \Naive\ performance only for the synthetic data sets (Figures \ref{fig:sparsesynstrongscale} and \ref{fig:densesynscaling}); the results for the real-world data sets are similar.
For the Sparse Synthetic data set, \Naive\ is $4.2\times$ slower than the fastest algorithm (\BPP) on 1536 processors; for the Dense Synthetic data set, \Naive\ is $1.6\times$ slower than the fastest algorithm (\MU) at that scale.
Nearly all of this slowdown is due to the communication costs of \Naive.
Theoretical and practical evidence supporting the first observation is also reported in our previous paper \cite{KBP16}.
However, we also note that \Naive\ is the fastest algorithm for the smallest $p$ for each problem, which is largely due to reduced MM time.
Each algorithm performs exactly the same number of flops per MM; the efficiency of \Naive\ for small $p$ is due to cache effects.
For example, for the Dense Synthetic problem on 96 processors, the output matrix of \Naive's MM fits in L2 cache, but the output matrix of \ParNMF's MM does not; these effects disappear as the $p$ increases.

\paragraph{Observation \ref{obs:2}}
Algorithms based on \ParNMF\ (\MU, \HALS, \BPP) scale well, up to over 1000 processors.
All algorithms' run times decrease as $p$ increases, with the exception of the Sparse Real World data set, in which case all algorithms slow down scaling from $p=864$ to $p=1536$ (we attribute this lack of scaling to load imbalance).
For sparse problems, comparing $p=16$ to $p=1536$ (a factor increase of 96), we observe speedups from \BPP\ of $59\times$ (synthetic) and $22\times$ (real world).
For dense problems, comparing $p=96$ to $p=1536$ (a factor increase of 16), \BPP's speedup is $12\times$ for both problems.
\MU\ and \HALS\ demonstrate similar scaling results.
For comparison, speedups for \Naive\ were $8\times$ and $3\times$ (sparse) and $6\times$ and $4\times$ (dense).

\paragraph{Observation \ref{obs:3}}
\MU, \HALS, and \BPP\ share all the same subroutines except those that are characterized as \NLS.
Considering only \NLS subroutines, \MU\ and \HALS\ require fewer operations than \BPP. However, \HALS\ has to make one additional communication for normalization of $\WW$. 
For small $p$, these cost differences are apparent in Figure \ref{fig:scaling}.
For example, for the sparse real world data set on 16 processors, \BPP's \NLS time is $16\times$ that of $\MU$, and the per iteration time differs by a factor of $4.5$.
However, as $p$ increases, the relative time spent in \NLS computations decreases, so the extra time taken by \BPP\ has less of an effect on the total per iteration time.
By contrast, for the dense real world data set on 1536 processors, \BPP\ spends a factor of 27 times more time in \NLS than \MU\ but only $11\%$ longer over the entire iteration.
For the synthetic data sets, \NLS takes $24\%$ (sparse) on 16 processors and $84\%$ (dense) on 96 processors, and that percentage drops to $11\%$ (sparse) and $15\%$ (dense) on 1536 processors.

These trends can also be seen theoretically (Table \ref{tab:costs}).
We expect local computations like MM, \NLS, and Gram to scale like $1/p$, assuming load balance is preserved.
If communication costs are dominated by the number of words being communicated (i.e., the communication is bandwidth bound), then we expect time spent in communication to scale like $1/\sqrt p$, and at least for dense problems, this scaling is the best possible.
Thus, communication costs will eventually dominate computation costs for all NMF problems, for sufficiently large $p$.
(Note that if communication is latency bound and proportional to the number of messages, then time spent communicating actually increases with $p$.)

The overall conclusion from this empirical and theoretical observation is that the extra per-iteration cost of \BPP\ over alternatives like \MU\ and \HALS\ decreases as the number of processors $p$ increases.
As shown in Section \ref{sec:convergence} the faster error reduction of \BPP\ typically reduces the overall time to solution compared with the alternatives even it requires more time for each iteration.
Our conclusion is that as we scale up $p$, this tradeoff is further relaxed so that \BPP\ becomes more and more advantageous for both quality and performance.

\begin{figure}

\naivetrue
\ksweepfalse

\begin{subfigure}[b]{\textwidth}
\centering
\begin{tikzpicture}
\renewcommand{\datafile}{sparsesynstrong_scale_pgf.dat}
\makeplot
\labels
\end{tikzpicture}
\subcaption{Sparse Synthetic}
\label{fig:sparsesynstrongscale}
\end{subfigure}

\begin{subfigure}[b]{\textwidth}
\centering
\begin{tikzpicture}
\renewcommand{\datafile}{densesynstrong_scale_pgf.dat}
\makeplot
\labels
\end{tikzpicture}
\subcaption{Dense Synthetic}
\label{fig:densesynscaling}
\end{subfigure}

\naivefalse

\begin{subfigure}[b]{\textwidth}
\centering
\renewcommand{\datafile}{sparserwstrong_scale_pgf.dat}
\begin{tikzpicture}
\makeplot
\labels
\end{tikzpicture}
\subcaption{Sparse Real World (webbase-1M)}
\label{fig:sparserwscaling}
\end{subfigure}

\begin{subfigure}[b]{\textwidth}
\centering
\begin{center}
\renewcommand{\datafile}{denserwstrong_scale_pgf.dat}
\begin{tikzpicture}
\makeplot
\labels
\end{tikzpicture}
\subcaption{Dense Real World (Video)}
\label{fig:denserwscaling}
\end{center}
\end{subfigure}

\caption{Strong scaling (varying $p$) with $k=50$ benchmarking per-iteration times.}
\label{fig:scaling}
\end{figure}
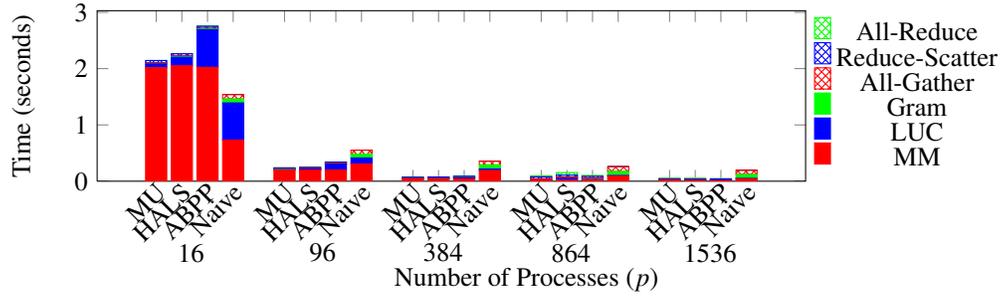
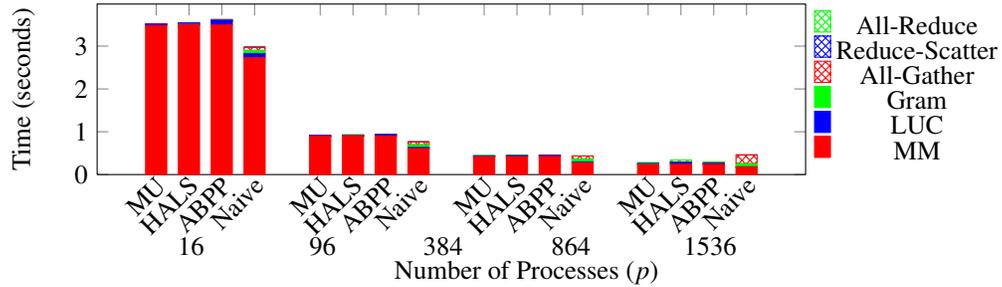
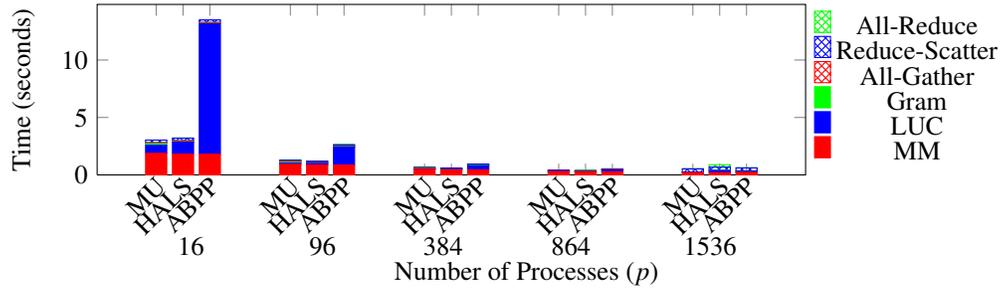
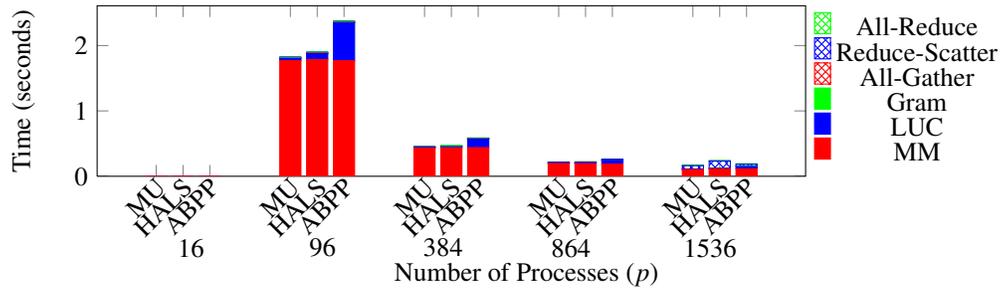

\subsubsection{Scaling \texorpdfstring{$k$}{k}}
\label{sec:ksweep}

Figure \ref{fig:ksweep} presents an experiment scaling up the low rank value $k$ from 10 to 50 with each of the four data sets.
In this experiment, for each data set and algorithm, the problem size is fixed and the number of processors is fixed to $p=864$.
As in Section \ref{sec:scaling}, we report the average per-iteration times.
We also omit \Naive\ data for the real world data sets to highlight the comparisons among \MU, \HALS, and \BPP.

We highlight two observations from these experiments:
\begin{enumerate}
	\item \label{obs:a} \Naive\ is plagued by communication time that increases linearly with $k$;
	\item \label{obs:b} \BPP's time increases more quickly with $k$ than those of \MU\ or \HALS;
\end{enumerate}

\paragraph{Observation \ref{obs:a}} 
We see from the synthetic data sets (Figures \ref{fig:sparsesynksweep} and \ref{fig:densesynksweep}) that the overall time of \Naive\ increases more rapidly with $k$ than any other algorithm and that the increase in time is due mainly to communication (All-Gather).
Table \ref{tab:costs} predicts that \Naive\ communication volume scales linearly with $k$, and we see that in practice the prediction is almost perfect with the synthetic problems.
This confirms that the communication is dominated by bandwidth costs and not latency costs (which are constant with respect to $k$).
We note that the communication cost of \ParNMF\ scales like $\sqrt k$, which is why we don't see as dramatic an increase in communication time for \MU, \HALS, or \BPP in Figure \ref{fig:ksweep}.

\paragraph{Observation \ref{obs:b}}
 Focusing attention on time spent in \NLS computations, we can compare how \MU, \HALS, and \BPP\ scale differently with $k$.
 We see a more rapid increase of \NLS time for \BPP\ than \MU\ or \HALS; this is expected because the \NLS computations unique to \BPP\ require between $O(k^3)$ and $O(k^4)$ operations (depending on the data) while the unique \NLS computations for \MU\ and \HALS\ are $O(k^2)$, with all other parameters fixed.
Thus, the extra per-iteration cost of \BPP\ increases with $k$, so the advantage of \BPP\ of better error reduction must also increase with $k$ for it to remain superior at large values of $k$.
We also note that although the number of operations within MM is $O(k)$, we do not observe much increase in time from $k=10$ to $k=50$; this is due to the improved efficiency of local MM for larger values of $k$.

\begin{figure}

\naivetrue
\ksweeptrue

\begin{subfigure}[b]{\textwidth}
\centering
\renewcommand{\datafile}{sparsesyn864_ksweep_scale_pgf.dat}
\begin{tikzpicture}
\makeplot
\labels
\end{tikzpicture}
\subcaption{Sparse Synthetic}
\label{fig:sparsesynksweep}
\end{subfigure}

\begin{subfigure}[b]{\textwidth}
\centering
\renewcommand{\datafile}{densesyn864_ksweep_scale_pgf.dat}
\begin{tikzpicture}
\makeplot
\labels
\end{tikzpicture}
\subcaption{Dense Synthetic}
\label{fig:densesynksweep}
\end{subfigure}

\naivefalse

\begin{subfigure}[b]{\textwidth}
\centering
\renewcommand{\datafile}{sparserwksweep_scale_pgf.dat}
\begin{tikzpicture}
\makeplot
\labels
\end{tikzpicture}
\subcaption{Sparse Real World (webbase-1M)}
\label{fig:sparserwksweep}
\end{subfigure}

\begin{subfigure}[b]{\textwidth}
\centering
\renewcommand{\datafile}{denserw_ksweep_scale_pgf.dat}
\begin{tikzpicture}
\makeplot
\labels
\end{tikzpicture}
\caption{Dense Real World (Video)}
\label{fig:denserwksweep}
\end{subfigure}

\caption{Varying low rank $k$ with $p=864$, benchmarking per-iteration times.}
\label{fig:ksweep}
\end{figure}
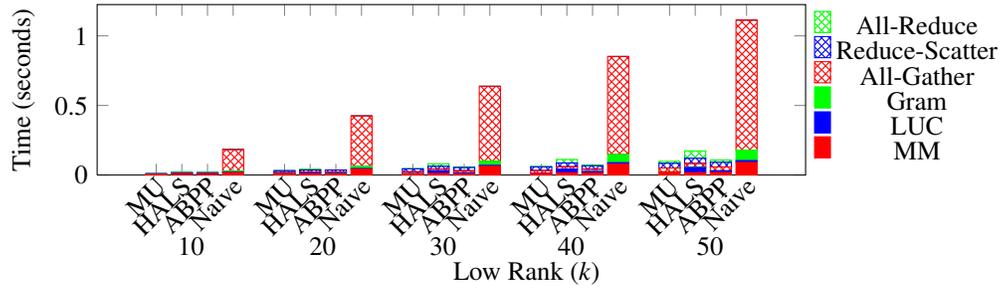
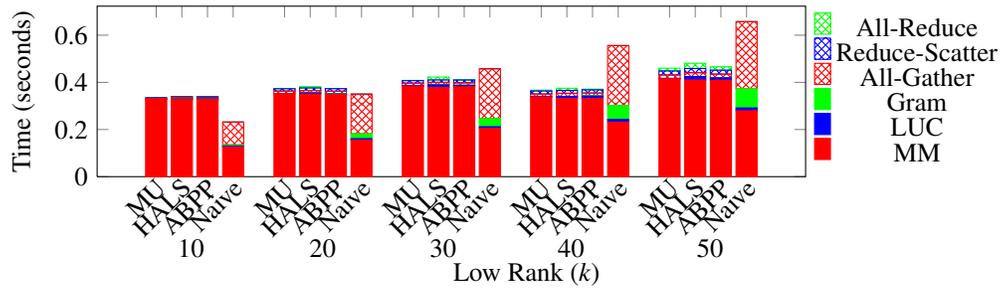
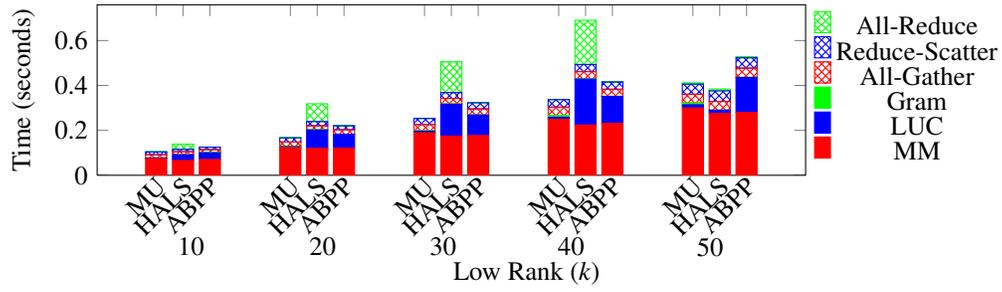
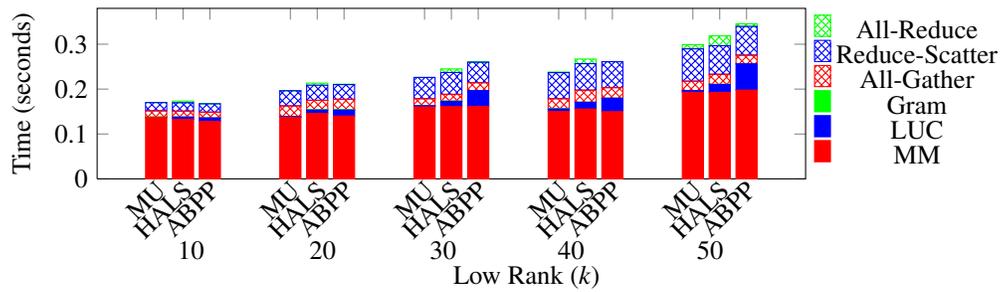

\subsubsection{Varying Processor Grid}
\label{sec:procgrid}

In this section we demonstrate the effect of the dimensions of the processor grid on per-iteration performance.
For a fixed total number of processors $p$, the communication cost of Algorithm \ref{alg:2D} varies with the choice of $p_r$ and $p_c$.
To minimize the amount of data communicated, the theoretical analysis suggests that the processor grid should be chosen to make the sizes of the local data matrix as square as possible.
This implies that if $m/p > n$, $p_r=p$ and $p_c=1$ is the optimal choice (a 1D processor grid); likewise if $n/p > m$ then a 1D processor grid with $p_r=1$ and $p_c=p$ is the optimal choice.
Otherwise, a 2D processor grid minimizes communication with $p_r \approx \sqrt{mp/n}$ and $p_c \approx \sqrt{np/m}$ (subject to integrality and $p_rp_c=p$).

Figure \ref{fig:procsweep} presents a benchmark of \BPP\ for the Sparse Synthetic data set for fixed values of $p$ and $k$.
We vary the processor grid dimensions from both 1D grids to the 2D grid that matches the theoretical optimum exactly.
Because the sizes of the Sparse Synthetic matrix are $172{,}800\times115{,}200$ and the number of processors is 1536, the theoretically optimal grid is $p_r = \sqrt{mp/n} = 48$ and $p_c = \sqrt{np/m} = 32$.
The experimental results confirm that this processor grid is optimal, and we see that the time spent communicating increases as the processor grid deviates from the optimum, with the 1D grids performing the worst.

\begin{figure}
\renewcommand{\datafile}{sparsesynpsweep_scale_pgf.dat}
\centering
\begin{tikzpicture}
\begin{axis}[
	ybar stacked,
	reverse legend,
	bar width=12pt,
	width=8cm, height=5cm,
	ylabel={Time (seconds)}, 
	x label style={yshift=-.5cm},
	ymin=0,
	symbolic x coords={1-1536,8-192,16-96,32-48,48-32,96-16,192-8,1536-1},
	xticklabels={$1{\times}1536$,$8{\times}192$,$16{\times}96$,$32{\times}48$,$48{\times}32$,$96{\times}16$,$192{\times}8$,$1536{\times}1$},
	xtick=data,
	xticklabel style={xshift=.15cm,rotate=45,anchor=east},
	xlabel={Processor Grid}, 
	legend style={draw=none,row sep=-0.1cm},
	legend style={at={(1,.5)},anchor=west}
]
	\setcolors
	\addplot table[x=pg, y expr=(\thisrow{mm}/(\minvalue*\numiterations))] {\datafile};
	\addplot table[x=pg, y expr=(\thisrow{nnls}/(\minvalue*\numiterations))] {\datafile};
	\addplot table[x=pg, y expr=(\thisrow{gram}/(\minvalue*\numiterations))] {\datafile};
	\addplot table[x=pg, y expr=(\thisrow{allgather}/(\minvalue*\numiterations))] {\datafile};
	\addplot table[x=pg, y expr=(\thisrow{reducescatter}/(\minvalue*\numiterations))] {\datafile};
	\addplot table[x=pg, y expr=(\thisrow{allreduce}/(\minvalue*\numiterations))] {\datafile};
	\legend{MM, \NLS, Gram, All-Gather, Reduce-Scatter, All-Reduce};
\end{axis}
\end{tikzpicture}
\caption{Tuning processor grid for \BPP\ on Sparse Synthetic data set with $p=1536$ and $k=50$.}
\label{fig:procsweep}
\end{figure}
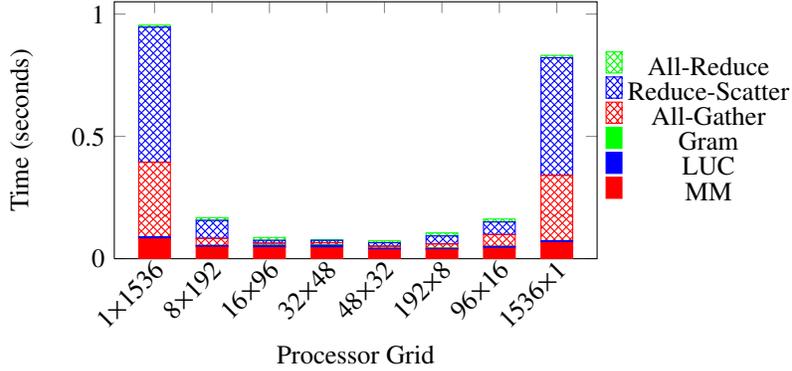

\subsubsection{Scaling up to Very Large Sparse Datasets}\label{sec:webbase-2001}

In this section, we test \ParNMF\ by scaling up the problem size.
While we've used \emph{webbase-1M} in previous experiments, we consider \emph{webbase-2001} in this section as it is the largest sparse data in University of Florida Sparse Matrix Collection \cite{DH11}.
The former dataset has about 1 million nodes and 3 million edges, whereas the latter dataset has over 100 million nodes and  1 billion edges (see Section \ref{sec:datasets} for more details).
Not only is the size of the input matrix increased by two orders of magnitude (because of the increase in the number of edges), but also the size of the output matrices is increased by two orders of magnitude (because of the increase in the number of nodes).

In fact, with a low rank of $k=50$, the size of the output matrices dominates that of the input matrix: $\WW$ and $\HH$ together require a total of 88 GB, while $\AA$ (stored in compressed column format) is only 16 GB.
At this scale, because each node (consisting of 16 cores) of Rhea has 128 GB of memory, multiple nodes are required to store the input and output matrices with room for other intermediate values.
As mentioned in Section \ref{sec:memory}, \ParNMF\ requires considerably more temporary memory than necessary when the output matrices require more memory than the input matrix.
While we were not limited by this property for the other sparse matrices, the \emph{webbase-2001} matrix dimensions are so large that we need the memories of tens of nodes to run the algorithm.
Thus, we report results only for the largest number of processors in our experiments: 1536 processors (96 nodes).
The extra temporary memory used by \ParNMF\ is a latency-minimizing optimization; the algorithm can be updated to avoid this extra memory cost using a blocked matrix multiplication algorithm.
The extra memory can be reduced to a negligible amount at the expense of more messages between processors and synchronizations across the parallel machine.
We have not yet implemented this update.



We present results for \emph{webbase-2001} in Figure \ref{fig:webbase2001}.
The timing results are consistent with the observations from other synthetic and real world sparse datasets as discussed in Section \ref{sec:scaling}, though the raw times are about 2 orders of magnitude larger, as expected. 
In the case of the error plot, as observed in other experiments, \BPP\ outperforms other algorithms; however  we see that \MU\ reduces error at a faster rate than \HALS\ in the first 30 iterations. 
At the 30th iteration, the error for \HALS\ was still improving at the third decimal, whereas \MU's was improving at the fourth decimal. 
We suspect that over a greater number of iterations the error of \HALS\ could become smaller than that of \MU, which would be more consistent with other datasets. 

\begin{figure}
\begin{subfigure}[t]{.5\textwidth}

\renewcommand{\datafile}{webbase118m_pgf.dat}
\centering
\begin{tikzpicture}
\begin{axis}[
	ybar stacked,
	reverse legend,
	bar width=16pt,
	width=5cm, height=5cm,
	ylabel={Time (seconds)}, 
	x label style={yshift=-.5cm},
	ymin=0,
	symbolic x coords={1536-0,1536-1,1536-2},
	xticklabels={MU,HALS,ABPP},
	xtick=data,
	legend style={draw=none,row sep=-0.1cm},
	legend style={at={(1,.5)},anchor=west}
]
	\setcolors
	\addplot table[x=algo, y expr=(\thisrow{mm})] {\datafile};
	\addplot table[x=algo, y expr=(\thisrow{nnls})] {\datafile};
	\addplot table[x=algo, y expr=(\thisrow{gram})] {\datafile};
	\addplot table[x=algo, y expr=(\thisrow{allgather})] {\datafile};
	\addplot table[x=algo, y expr=(\thisrow{reducescatter})] {\datafile};
	\addplot table[x=algo, y expr=(\thisrow{allreduce})] {\datafile};
	\legend{MM, \NLS, Gram, All-Gather, Reduce-Scatter, All-Reduce};
\end{axis}
\end{tikzpicture}
\subcaption{Time}
\end{subfigure}
\begin{subfigure}[t]{.5\textwidth}
    \centering
\begin{tikzpicture}
\renewcommand{\datafile}{webbase118merr.dat}
\begin{axis}[xlabel=Iterations, ylabel=Rel.~Error for $k{=}50$,width=5cm, height=5cm]
\addplot [green,very thick] table [x={Iterations}, y={MU}] {\datafile};
\addplot [red,very thick] table [x={Iterations}, y={HALS}] {\datafile};
\addplot [blue,very thick] table [x={Iterations}, y={ANLS-BPP}] {\datafile};
\legend{MU,HALS,ABPP}
\end{axis}
\end{tikzpicture}
\subcaption{Error}
\end{subfigure}
\caption{NMF comparison on \emph{webbase-2001} for $k{=}50$ on 1536 processors.}
\label{fig:webbase2001}
\end{figure}
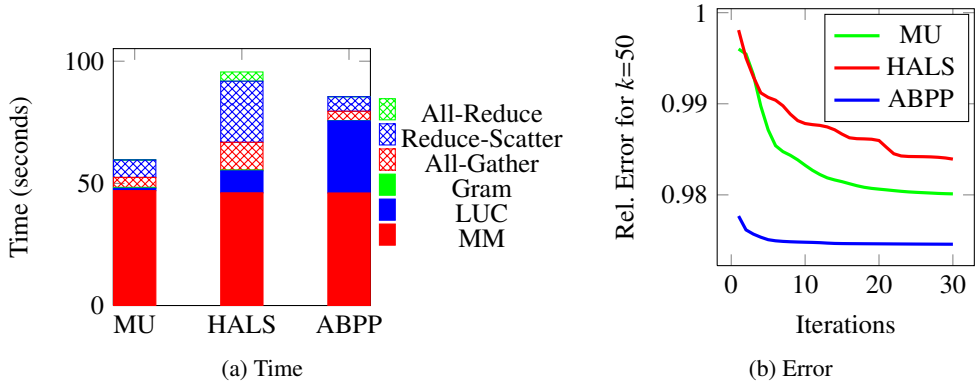

\subsection{Interpretation of Results}

In this section, we present results from two of the real world datasets. 
The first example shows an image processing example of background separation and moving object detection in surveillance video data, and the second example shows topic modeling output on the \emph{stack exchange} text dataset. The details of these datasets are presented in Section \ref{sec:datasets}. 
While the literature covers more detail about fine tuning NMF and different NMF variants for higher quality results on these two tasks  \cite{ZT2011,B2014,AGHKT2014,KJJCH2015}, our main focus is to show how quickly we can produce a baseline NMF output and its real world interpretation. 

\subsubsection{Moving Object Detection of Surveillance Video Data}

As explained in the Section \ref{sec:datasets}, we processed 12 minutes video that is captured from a 
busy junction in Georgia Tech to separate the background and moving objects from this video. 
In Figure \ref{fig:videoresults} we present some sample frames to compare the input image with the separated background and moving objects.
The background are the results of the low rank approximation 
$\hat{\AA}=\WW\HH$ output yielded from our \ParNMF\ algorithm and the moving objects are given by $\AA-\hat{\AA}$. 
We can clearly see the background remains static and the moving objects (e.g., cars) are visible. 

\newcommand{\wdth}{1.65in}
\newcommand{\hght}{1.0725in}
\begin{figure}
\centering
\begin{tabular}{ccc}
Input Frame($\AA$) & Background ($\WW\HH$) & Moving Object $\AA-\WW\HH$ \\
\includegraphics[height=\hght,width=\wdth]{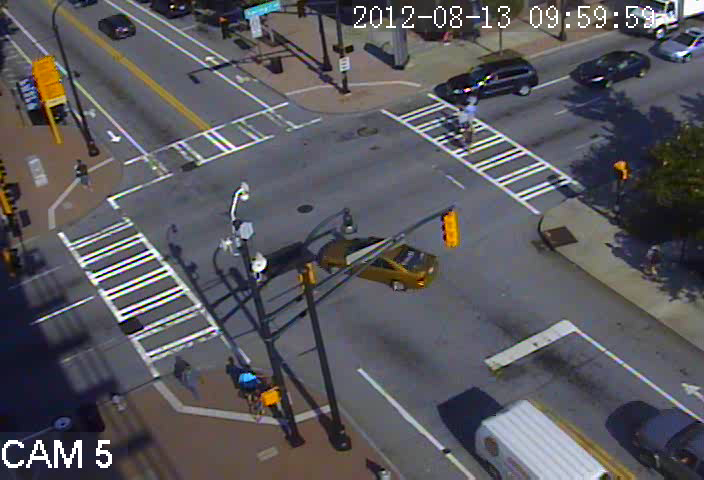} &  
\includegraphics[height=\hght,width=\wdth]{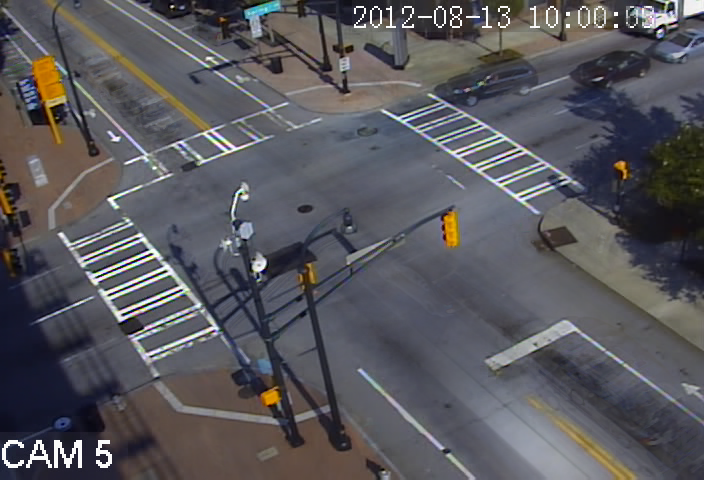}&
\includegraphics[height=\hght,width=\wdth]{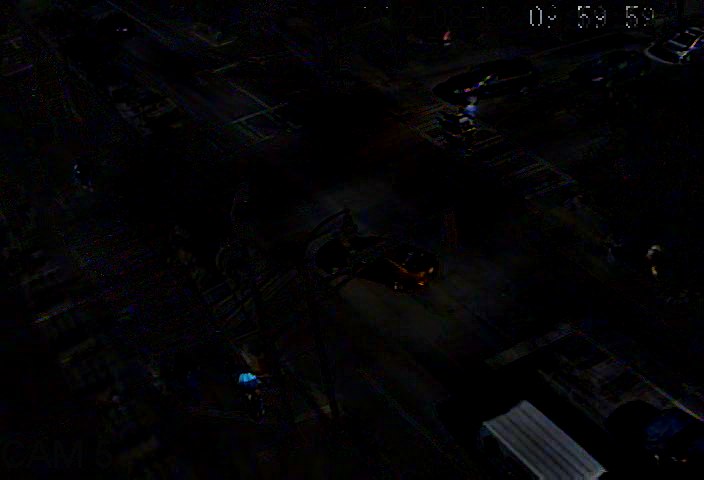} \\
\includegraphics[height=\hght,width=\wdth]{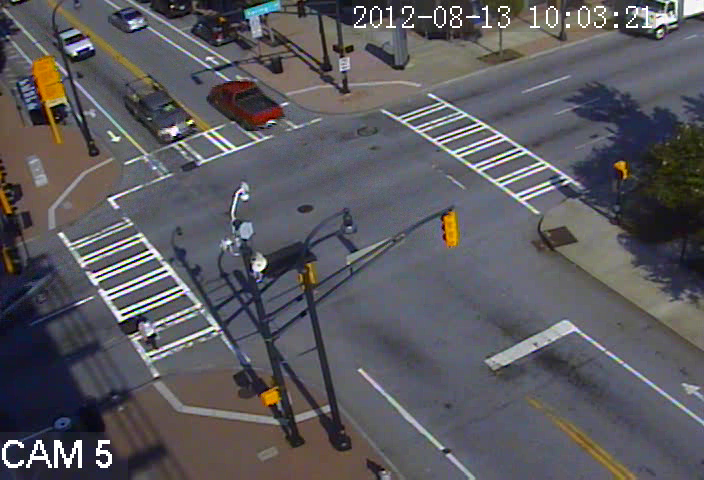} &  
\includegraphics[height=\hght,width=\wdth]{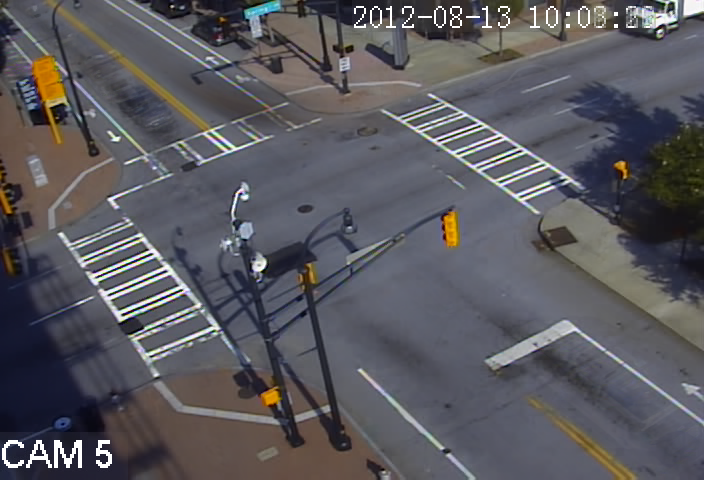}&
\includegraphics[height=\hght,width=\wdth]{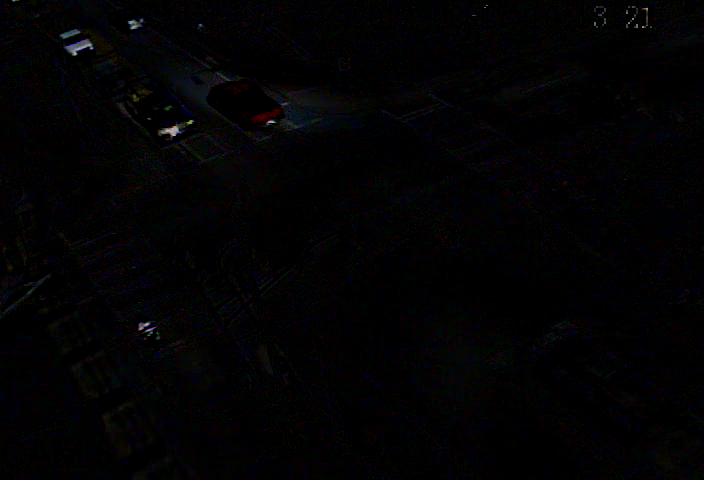} \\
\includegraphics[height=\hght,width=\wdth]{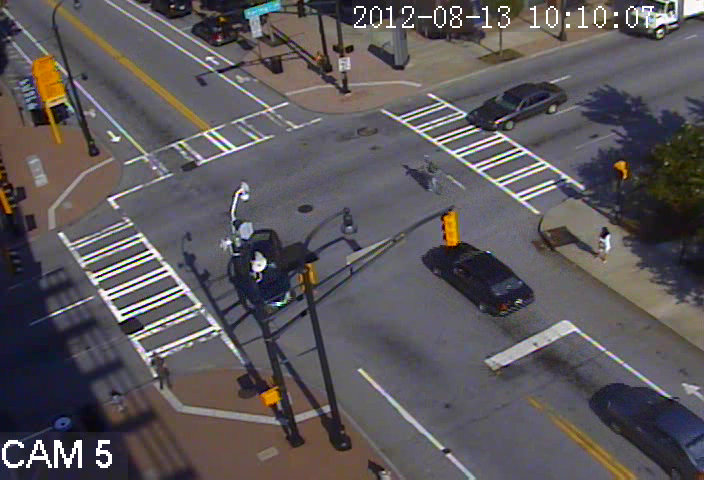} &  
\includegraphics[height=\hght,width=\wdth]{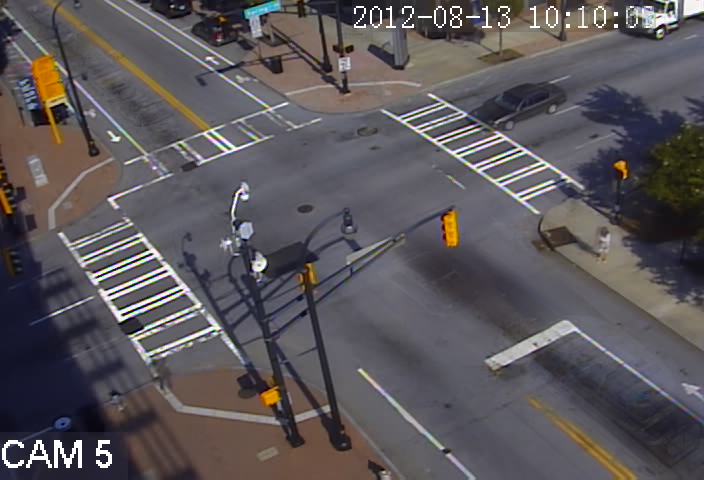}&
\includegraphics[height=\hght,width=\wdth]{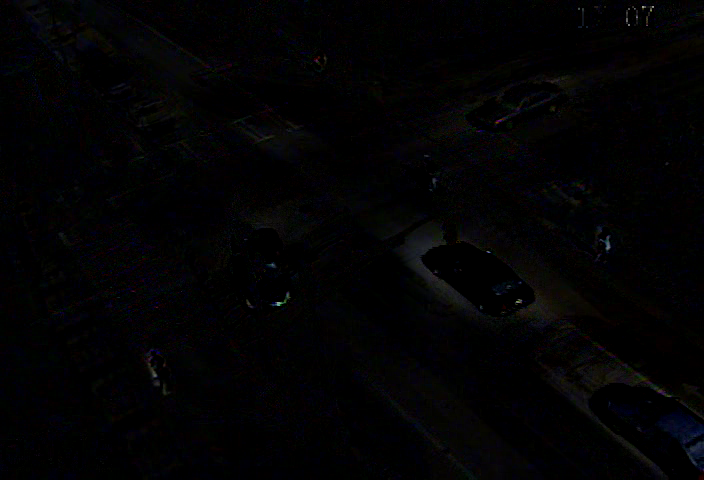} \\
\end{tabular}
\caption{Moving object detection for video data using NMF.  Each row of images corresponds to a particular frame in the video.  The left column is the original frame, the middle column is the reconstructed frame from the low-rank approximation (which captures the background), and the right column is the difference (which captures the moving objects).}
\label{fig:videoresults}
\end{figure}

\subsubsection{Topic Modeling of Stack Exchange Data}

We downloaded the latest Stack Overflow  dump from its archive on 28-Jul-2016. The details of 
the preprocessing and the sparse matrix generation are explained in Section \ref{sec:datasets}. We ran 
our \ParNMF\ algorithm on this dataset, which has nearly 12 million questions from the Stack Overflow site (under Stack Exchange) to produce 50 topics. 
The matrix $\WW$ can be interpreted as {\em vocabulary-topic} distribution and the 
$\HH$ as {\em topic-document} distribution.  We took the top 5 words for each of the 50 topics and 
present them in Table \ref{tab:stackexchangetopics}. Typically a good topic generation satisfies 
properties such as (a) finding discriminative rather than common words -- capturing words that can provide 
some information; (b) finding different topics -- the similarity between different topics should be low; 
(c) coherence - all the words that belong to one topic should be coherent.  There are some topic quality 
metrics \cite{NLGB2010} that capture the usefulness of topic generation algorithm.  We can  see
 NMF generated generally high-quality and coherent topics. 
 Also, each of the topics are from different domains such as 
 databases, C/C++ programming, Java programming, and web technologies like PHP and HTML.

\begin{table}
\begin{center}
\footnotesize
\begin {tabular}{|ccccc||>{\columncolor [gray]{.8}}c>{\columncolor [gray]{.8}}c>{\columncolor [gray]{.8}}c>{\columncolor [gray]{.8}}c>{\columncolor [gray]{.8}}c|}%
\toprule 
\multicolumn{5}{|c||}{Top Keywords from Topics 1-25} & \multicolumn{5}{|>{\columncolor [gray]{0.8}}c|}{Top Keywords from Topics 26-50} \\ 
word1&word2&word3&word4&word5&word1&word2&word3&word4&word5\\\midrule %
refer&undefin&const&key&compil&echo&type=text&php&form&result\\%
text&field&box&word&static&test&perform&fail&unit&result\\%
imag&src&descript&alt=ent&size&tabl&key&queri&databas&insert\\%
button&click&event&form&add&user&email&usernam&login&log\\%
creat&bean&add&databas&except&data&json&store&read&databas\\%
string&static&final&catch&url&page&load&content&url&link\\%
width&height&color&left&display&privat&static&final&import&float\\%
app&applic&servic&thread&work&row&column&date&cell&valu\\%
ipsum&lorem&dolor&sit&amet&line&import&command&print&recent\\%
node&list&root&err&element&var&map&marker&match&url\\%
0x00&0xff&byte&0x01&0xc0&server&connect&client&messag&request\\%
file&directori&read&open&upload&number&byte&size&print&input\\%
function&call&event&work&variabl&object&properti&json&instanc&list\\%
int&char&const&static&doubl&array&element&valu&key&index\\%
public&overrid&virtual&static&extend&main&thread&program&frame&cout\\%
return&param&result&def&boolean&type&field&properti&argument&resolv\\%
info&thread&start&map&servic&select&item&queri&join&list\\%
error&syntax&found&symbol&fail&sourc&target&except&java&fail\\%
set&properti&virtual&default&updat&instal&version&packag&err&default\\%
case&break&switch&default&cout&code&work&problem&chang&write\\%
method&call&except&static&todo&void&overrid&protect&catch&extend\\%
href&nofollow&src&link&work&true&requir&boolean&option&valid\\%
end&def&dim&begin&properti&find&project&import&warn&referenc\\%
debug&request&filter&match&found&view&control&item&overrid&posit\\%
fals&boolean&fix&bool&autoincr&null&default&key&int(11&primari\\\bottomrule %
\end {tabular}%

\normalsize
\end{center}
\caption{Top 5 words of 50 topics from {\em Stack Exchange} data set.}
\label{tab:stackexchangetopics}
\end{table}%
\section{Conclusion}\label{sec:conclusion}

In this paper, we propose a high-performance distributed-memory parallel framework for NMF algorithms that iteratively update the low rank factors in an alternating fashion.   
Our parallel algorithm is designed to avoid communication overheads and scales well to over 1500 cores. 
The framework is flexible, being (a) expressive enough to leverage many different NMF algorithms and (b) efficient for both sparse and dense matrices of sizes that span from a few hundreds to hundreds of millions.  
Our open-source software implementation is available for download.

For solving data mining problems at today's scale, parallel computation and distributed-memory systems are becoming prerequisites.
We argue in this paper that by using techniques from high-performance computing, the computations for NMF can be performed very efficiently.
Our framework allows for the HPC techniques (efficient matrix multiplication) to be separated from the data mining techniques (choice of NMF algorithm), and we compare data mining techniques at large scale, in terms of data sizes and number of processors.
One conclusion we draw from the empirical and theoretical observations is that the extra per-iteration cost of \BPP\ over alternatives like \MU\ and \HALS\ decreases as the number of processors $p$ increases, making \BPP\ more advantageous in terms of both quality and performance at larger scales.
By reporting time breakdowns that separate local computation from interprocessor communication, we also see that our efficient algorithm prevents communication from bottlenecking the overall computation; our comparison with a naive approach shows that communication can easily dominate the running time of each iteration. 


In future work, we would like to extend \ParNMF\  algorithm to dense and sparse tensors, computing the CANDECOMP/PARAFAC decomposition in parallel with non-negativity constraints on the factor matrices.
We plan on extending our software to include more NMF algorithms that fit the AU-NMF framework; these can be used for both matrices and tensors. 
We would also like to explore more intelligent distributions of sparse matrices: while our 2D distribution is based on evenly dividing rows and columns, it does not necessarily load balance the nonzeros of the matrix, which can lead to load imbalance in matrix multiplications.
We are interested in using graph and hypergraph partitioning techniques to load balance the memory and computation while at the same time reducing communication costs as much as possible.

\acks
\small

This manuscript has been co-authored by UT-Battelle, LLC under Contract No. DE-AC05-00OR22725 with the U.S. Department of Energy.  This project was partially funded by the Laboratory Director's Research and Development fund. This research used resources of the Oak Ridge Leadership Computing Facility at the Oak Ridge National Laboratory, which is supported by the Office of Science of the U.S. Department of Energy.


Also, partial funding for this work was provided by AFOSR Grant FA9550-13-1-0100, National Science Foundation (NSF) grants IIS-1348152 and ACI-1338745, Defense Advanced Research Projects Agency (DARPA) XDATA program grant FA8750-12-2-0309.

The United States Government retains and the publisher, by accepting the article for publication, acknowledges that the United States Government retains a non-exclusive, paid-up, irrevocable, world-wide license to publish or reproduce the published form of this manuscript, or allow others to do so, for United States Government purposes. The Department of Energy will provide public access to these results of federally sponsored research in accordance with the DOE Public Access Plan (\url{http://energy.gov/downloads/doepublic-access-plan}). 

 Any opinions, findings and conclusions or recommendations expressed in this material are those of the authors and do not necessarily reflect the views of the USDOE, NERSC, AFOSR, NSF or DARPA.
 
 \normalsize

\vspace{-0.1in}

\bibliographystyle{ACM-Reference-Format-Journals}
\bibliography{paper}



\end{document}